\documentclass[12pt]{article}

\usepackage[a4paper,margin=1in]{geometry}
\usepackage{setspace}
\usepackage{amsmath,amssymb,amsthm,mathtools}
\usepackage{bm}

\usepackage{graphicx}
\usepackage{subcaption} 
\usepackage{booktabs}
\usepackage{float}

\usepackage[hidelinks]{hyperref}
\usepackage[nameinlink,noabbrev]{cleveref} 
\usepackage{natbib}
\usepackage{microtype}

\usepackage{tikz}
\usetikzlibrary{patterns}
\usepackage{comment}
\usepackage{siunitx}
\usepackage{soul}
\usepackage{hyperref}
\usepackage{dsfont}
\usepackage{verbatim}
\usepackage{subcaption}
\usepackage{newtxtext}
\usepackage[subscriptcorrection]{newtxmath}

\graphicspath{{./art/}}

\DeclareMathOperator{\lrs}{\lambda_{\mbox{\scriptsize LR}}}

\theoremstyle{plain}
\newtheorem{theorem}{Theorem}[section]

\newtheorem{lemma}[theorem]{Lemma}
\newtheorem{corollary}[theorem]{Corollary}

\theoremstyle{definition}
\newtheorem{definition}[theorem]{Definition}

\theoremstyle{remark}
\newtheorem{remark}[theorem]{Remark}

\title{Asymptotic distribution of the likelihood ratio test statistic with inequality-constrained nuisance parameters}
\author{Clara Bertinelli Salucci \\[1em] Faculty of Mathematics and Natural Science, Department of Mathematics, \\ University of Oslo, 35 Moltke Moes Vei, Oslo 0316, Norway \\
\href{mailto:clarabe@math.uio.no}{clarabe@math.uio.no}}

\date{} 

\begin{document}
\maketitle

\begin{abstract}
\noindent The asymptotic distribution of the likelihood-ratio statistic for testing
parameters on the boundary is well known to be a chi-squared mixture. The mixture weights have been shown to correspond to the intrinsic volumes of an associated tangent cone, unifying a wide range of previously isolated special cases.  While the weights are fully understood for an arbitrary number of parameters of interest on the boundary, much less is known when nuisance parameters are also constrained to the boundary, a situation that
frequently arises in applications. We provide the first general characterization of the asymptotic distribution of the likelihood-ratio test statistic when both the number of parameters of interest and the number of nuisance parameters on the boundary are arbitrary.
We analyze how the cone geometry changes when moving from a problem with
$K$ parameters of interest on the boundary to one with $K-m$ parameters of interest and $m$ nuisances.  In the orthogonal case we show that
the resulting change in the chi-bar weights admits a closed-form
difference pattern that redistributes probability mass across adjacent
degrees of freedom, and that this pattern remains the dominant component
of the weight shift under arbitrary covariance structures when the
nuisance vector is one-dimensional.  For a generic number of nuisance
parameters, we introduce a new rank-based aggregation of intrinsic
volumes that yields an accurate
approximation of the mixture weights.  Comprehensive simulations support the theory and demonstrate the accuracy of the proposed approximation.
\end{abstract}


\section{Introduction}

The study of the asymptotic distribution of likelihood ratio test statistics $\lrs$ under
inequality constraints has its roots in the work of \citet{Chernoff1954}, who
showed that, under mild regularity conditions, the statistic for testing a simple null hypothesis against a convex cone alternative can be represented as the squared distance of a Gaussian vector to a closed convex cone. The asymptotic distribution is then a finite mixture of
chi-squared contributions, typically denoted by $\bar\chi^2$ (chi-bar-squared) \citep{Kudo1963}. Early explicit closed-form expressions for the weights in special polyhedral configurations, such as one-sided or orthant-type alternatives, appeared in subsequent contributions, e.g. \citep{Bartholomew}, where the mixture weights were obtained via direct calculations involving multivariate normal probabilities. Specific formulas for a few configurations were later given by \citet{SelfLiang1987}, who obtained the corresponding weights via purely geometric arguments involving the structure of the underlying tangent cones.

Up to that point, the available results were largely confined to isolated special cases. A major conceptual shift occurred with the seminal work of \citet{Shapiro1985, Shapiro1988}, who showed that the mixture weights admit a general geometric expression as Gaussian probabilities attached to the faces of a tangent cone and its polar, an insight that unified the diverse collection of isolated examples into a coherent geometric picture. Subsequent work has considerably deepened the geometric and probabilistic foundations of this theory. \citet{kuriki_takemura} linked Shapiro's facewise Gaussian-probability representation to the intrinsic volumes of polyhedral cones, providing the first explicit identification of chi-bar weights with conic intrinsic volumes, and further extended the theory to smooth or piecewise smooth cones. A major conceptual advance was made by \citet{amelunxen2014}, who established fundamental probabilistic identities linking intrinsic volumes to Gaussian projections onto convex cones. Building on these ideas, \citet{AmelunxenLotz2017} developed a detailed combinatorial theory for intrinsic volumes of polyhedral cones, proving explicit facewise decompositions that match the Gaussian internal and external angles appearing already in Shapiro's work.
Together, these developments transformed Shapiro's original insight into a broad and versatile geometric framework for analysing boundary problems in the context of likelihood ratio tests through the lens of convex conic geometry.

Despite this progress, the role of nuisance parameters has remained largely absent from this unified geometric treatment; yet in most applied settings -- from particle physics \citep{hep} and cosmology \citep{cosmology} to econometrics \citep{econometrics} and biostatistics \citep{biostat} -- such parameters, often abundant, are unavoidable and ought to be incorporated in any truly unified theory. To our knowledge, only \citet{SelfLiang1987} and \citet{KS} investigate the implications of having nuisance parameters on the boundary, but their analysis is limited to at most one-two nuisance parameters together with a single parameter of interest also on the boundary, and it does not return to the intrinsic-volume or angle-based representations mentioned above.

In this work we expand the general scope implied by Shapiro’s original contribution by making three contributions. First, we derive an exact
characterization of the effect of demoting parameters of interest to boundary nuisance parameters in the orthogonal case, starting from an arbitrary number of parameters of interest. In particular, we obtain a closed-form description of the resulting $\bar{\chi}^2$ weights, which exhibit a symmetric redistribution of probability mass across adjacent components. This yields a precise geometric representation of how the null cone changes when one or, more generally, any
number $m<K$ of parameters on the boundary are treated as nuisance.

Second, we establish a stability property of the weights under covariance
perturbations. For any polyhedral cone, we show that the orthogonal difference pattern remains the leading-order term when the covariance departs from identity under positive correlation. The discrepancy between the correlated and orthogonal weights is bounded by a constant multiple of the cone dimensionality and by an anisotropy index measuring the deviation of the transformed cone from orthogonality. This motivates the use of the orthogonal difference pattern to approximate the $\bar{\chi}^2$ weights in non-orthogonal settings when a parameter is demoted to nuisance.

Third, since this approximation may deteriorate when multiple parameters are demoted to nuisance, as the geometric deformation becomes dominated by the interaction between the nuisance coordinates and the covariance
structure, we introduce a new rank-based aggregation of intrinsic volumes obtained via Schur complements of the Fisher information matrix. This construction yields accurate and efficient approximate $\bar{\chi}^2$ weights for arbitrary numbers of boundary nuisances, with particularly strong performance when the number of parameters of interest is moderate (approximately up to $6$-$7$).

All results, both theoretical developments and numerical investigations, are developed within the locally asymptotically normal framework. Comprehensive simulations confirm the exactness of the lemmas derived in the orthogonal case and demonstrate the effectiveness of the proposed approximations in the presence of correlated parameters.

\section{Background and notation} \label{sec:background} 

Let $\{f(x;\theta): \theta \in \mathbb{R}^K\}$ be a regular parametric family, and let $(X_1,\ldots,X_n)$ be i.i.d.\ with density $f( \,\, \cdot \,\,;\theta_0)$ for some true parameter value $\theta_0 \in \mathbb{R}^K$. Let $P_\theta$ denote the distribution of a single observation with density $f(\,\, \cdot \,\,;\theta)$, and let $P_\theta^n$ be the joint distribution of the sample. Denote by \( \ell_n(\theta)=\sum_{i=1}^n \log f(X_i;\theta) \) the log-likelihood function, let \( \ell'_{\theta_0}(X) = \frac{\partial}{\partial \theta}\,\ell(X;\theta)\big|_{\theta=\theta_0} \) be the score function at $\theta_0$ for a single observation, and let \( I(\theta_0) = -\,\mathbb{E}_{\theta_0}\!\left[ \frac{\partial^2}{\partial \theta\,\partial \theta^\top} \ell(X;\theta)\vert_{\theta=\theta_0} \right] \) be the expected Fisher information matrix at $\theta_0$. Under standard differentiability, identifiability, and Fisher-information regularity conditions, the model is locally asymptotically normal (LAN) at $\theta_0$ \citep{LeCam1960,vanDerVaart1998}. In particular, letting \(W_n=\frac{1}{\sqrt{n}} \sum_{i=1}^n \ell'_{\theta_0}(X_i), \) we have $W_n \xrightarrow{d} W \sim N_K\bigl(0,I(\theta_0)\bigr)$ under $P_{\theta_0}$, and uniformly for $h$ in compact subsets of $\mathbb{R}^K$, 
\begin{equation} \ell_n\!\left(\theta_0 + \frac{h}{\sqrt{n}}\right) - \ell_n(\theta_0) = h^\top W_n - \tfrac{1}{2} h^\top I(\theta_0) h + o_{P_{\theta_0}}(1). \label{eq:LAN} \end{equation} 

\noindent Eq.~\eqref{eq:LAN} shows that, under the local reparametrization $\theta=\theta_0 + h/\sqrt{n}$, the log-likelihood process converges to that of a Gaussian shift experiment with log-likelihood ratio \( \Lambda(h;W)=h^\top W - \tfrac{1}{2} h^\top I(\theta_0) h, \) so that inference in a neighbourhood of $\theta_0$ can be studied via the geometry of this limiting quadratic form. The LAN expansion also yields the asymptotic normality of the maximum likelihood estimator (MLE) $\hat\theta_n$, as \( Z_n := \sqrt{n}\,(\hat\theta_n-\theta_0) = I^{-1}(\theta_0)\,W_n + o_{P_{\theta_0}}(1)\;\xrightarrow{d}\; Z \sim N_K\!\bigl(0, I^{-1}(\theta_0)\bigr). \) Equivalently, expressing the limiting log-likelihood ratio in terms of
$Z$, we may write
\(
\Lambda(h;Z)=h^\top I(\theta_0)Z-\tfrac12 h^\top I(\theta_0)h.
\) 
Therefore, following \citet{SelfLiang1987}, this boundary problem is asymptotically equivalent to estimating the restricted mean of a single Gaussian observation \( Z \sim N_K\!\bigl(0,I^{-1}(\theta_0)\bigr) \), with the MLE given by the projection of $Z$ onto the appropriate cone of admissible means. 

\medskip 

Suppose that the parameter space is subject to inequality constraints of the form $\theta_i \ge 0$ for $i = 1, \ldots, K$, and that under the null we have $\theta_{0,i}=0$ for all $i=1,\ldots,K$. In the local coordinate $h = \sqrt{n}(\theta - \theta_0)$, these constraints define a convex cone 
\( C = \{\,h \in \mathbb{R}^K : h_i \ge 0 \text{ for all } i\,\}\). Let $C_0 \subseteq C$ denote the cone associated with the null hypothesis $\theta_i = 0$ for $i = 1, \ldots, K$, which, in this fully constrained boundary case, reduces to $C_0 = \{0\}$. As mentioned, in the limiting Gaussian experiment the MLE under $H_1$
(resp.\ $H_0$) is obtained by maximizing the limiting log-likelihood
ratio $\Lambda(h;Z)$ over $h \in C$ (resp.\ $h \in C_0$), which is
equivalently given by the $I(\theta_0)$-orthogonal projection of $Z$
onto $C$ (resp.\ $C_0$). Indeed, since
\[
\Lambda(h;Z)
=
h^\top I(\theta_0) Z - \tfrac12 h^\top I(\theta_0) h
=
\tfrac12 \|Z\|_{I(\theta_0)}^2
-
\tfrac12 \|Z-h\|_{I(\theta_0)}^2,
\]
we have
\begin{align*}
\sup_{h \in C} \Lambda(h;Z)
=
\tfrac12 \|Z\|_{I(\theta_0)}^2
-
\tfrac12 &\inf_{h \in C} \|Z-h\|_{I(\theta_0)}^2
=
\tfrac12 \|P_C(Z)\|_{I(\theta_0)}^2 \\ \sup_{h \in C_0} \Lambda(h;Z)
&=
\tfrac12 \|P_{C_0}(Z)\|_{I(\theta_0)}^2,
\end{align*}

\noindent where $P_C(Z)$ denotes the $I(\theta_0)$-orthogonal projection of $Z$
onto $C$, and $\|x\|_{I(\theta_0)}^2 = x^\top I(\theta_0) x$.
The likelihood-ratio test statistic therefore admits the asymptotic
representation
\begin{equation}
\lrs
=
2\Big\{
\sup_{h \in C} \Lambda(h;Z)
-
\sup_{h \in C_0} \Lambda(h;Z)
\Big\}
=
\|P_C(Z)\|_{I(\theta_0)}^2
-
\|P_{C_0}(Z)\|_{I(\theta_0)}^2,
\label{eq:projectionI}
\end{equation}
which coincides with Equation~(3.1) of \citet{SelfLiang1987} written in projection form.

For analytical convenience, Self \& Liang perform a linear transformation that converts the $I(\theta_0)$-inner product into the ordinary Euclidean one. Let \( I(\theta_0) = P D P^{\top} \) be the spectral decomposition of the information matrix, with $P$ orthogonal matrix and $D=\mathrm{diag}(d_1,\ldots,d_K)$ a positive diagonal matrix of eigenvalues. Define the transformed variables 
\begin{equation} \tilde Z = D^{1/2} P^{\top} Z, \qquad \tilde h = D^{1/2} P^{\top} h, \label{eq:whitening} \end{equation}
\noindent so that $\tilde Z \sim N_K(0, \mathbb{I}_K)$, where $\mathbb{I}_K$ is the $K$-dimensional identity matrix. This linear mapping makes the inner product become Euclidean: \[ \langle x, y \rangle_{I(\theta_0)} = x^{\top} I(\theta_0) y = (D^{1/2} P^{\top} x)^{\top}(D^{1/2} P^{\top} y) = \langle \tilde x, \tilde y \rangle. \] \noindent The corresponding transformed cones are 

\begin{equation} \tilde C = \{ D^{1/2} P^{\top} h : h \in C \}, \qquad \tilde C_0 = \{ D^{1/2} P^{\top} h : h \in C_0 \}, \label{eq:transformed_cones} \end{equation} 

\noindent and Eq.~\eqref{eq:projectionI} can be rewritten as 
\( \lrs =\|P_{\tilde C}(\tilde Z)\|^2-\|P_{\tilde C_0}(\tilde Z)\|^2 \), where $\|\cdot\|$ denotes the standard Euclidean norm. The random vector $\tilde Z$ now has independent standard normal components, while all information about correlations and relative scales has been absorbed into the geometry of the transformed cones. 

\medskip 

In the case of $K$ parameters of interest (hereafter PoIs), the null distribution of $\lrs$ is a finite mixture of chi-square distributions with different degrees of freedom, 
\begin{equation} 
\lrs\ \sim\ \sum_{j=0}^{K}w_j\,\chi^2_j, \qquad w_j\ge0,\quad \sum_{j=0}^{K}w_j=1, \nonumber 
\end{equation}

\noindent i.e., the $\bar{\chi}^2$ distribution \citep{Kudo1963}. The mixture weights $w_j=w_j(\tilde C, \tilde C_0)$ depend only on the geometry of the alternative and null cones; for the case of $K$ PoIs on the boundary, with $\tilde C_0 = \{0\}$, they depend on the alternative cone only and can be expressed in terms of Gaussian conic angles as follows \citep{Shapiro1985, Shapiro1988, kuriki_takemura}. For each face $F$ of $\tilde C$ of dimension~$j$, let $\alpha(F)$ and $\beta(F)$ denote its internal and external Gaussian angles; then, we can express the weights as 
\begin{equation} 
w_j(\tilde C) =\sum_{\substack{F\subseteq\tilde C\\\dim(F)=j}}\alpha(F)\,\beta(F), \label{eq:angles}
\end{equation} 

\noindent where the summation runs, for each $j = 0, \dots, K$ over the faces of $\tilde C$ with dimension $j$. As mentioned before, these quantities coincide with the conic intrinsic volumes of~$\tilde C$ \citep{amelunxen2014,AmelunxenLotz2017}. In the special case where $\tilde C=\mathbb{R}_+^K$, the parameters are orthogonal, the angles in Eq.~\eqref{eq:angles} factorize and the weights reduce to the binomial form \citep{Shapiro1985, Shapiro1988} 

\begin{equation} 
w_j(\mathbb{R}_+^K)=2^{-K}\binom{K}{j}, 
\label{eq:orthant} 
\end{equation} 

\noindent reflecting that each coordinate of $\tilde Z$ independently falls in the positive half-line with probability $1/2$. When the parameters are not orthogonal, the cone $\tilde C$ is an oblique linear image of the orthant and the angles in Eq.~\eqref{eq:angles} no longer factorize. Writing the linear mapping as in Eq.\eqref{eq:whitening} and \eqref{eq:transformed_cones}, the alternative cone is $\tilde C = D^{1/2}P^{\top}\mathbb{R}_+^K = \{ A u : u \in \mathbb{R}_+^K\}$ with generator matrix $A = [a_1,\ldots,a_K] = D^{1/2}P^{\top}$. Let $G$ be the (Euclidean) Gram matrix of the generators, $G=(a_i^{\top}a_j)_{1\le i,j\le K}=A^{\top}A$, and let $H=G^{-1}$. For any subset of indices $\mathcal{S}\subseteq\{1,\ldots,K\}$, we denote by \( F_{\mathcal{S}} = \{ \sum_{i\in\mathcal{S}} \lambda_i a_i : \lambda_i \ge 0 \} \) the face of $\tilde C$ generated by the rays $\{a_i : i\in\mathcal{S}\}$, which has dimension $|\mathcal{S}|=j$. The internal and external Gaussian angles admit the following probabilistic representations in terms of orthant probabilities of centred normal distributions with covariance matrices built from principal submatrices of $H$ and of the polar Gram matrix $\widehat G$, see e.g. \citep{Shapiro1985,Shapiro1988,kuriki_takemura}, 
\[ \alpha(F_{\mathcal{S}}) \;=\; \Pr\!\big\{\,N_{|\mathcal{S}|}(0,\;H_{\mathcal{S}\mathcal{S}}) \in \mathbb{R}_+^{|\mathcal{S}|}\,\big\}, \qquad \beta(F_{\mathcal{S}}) \;=\; \Pr\!\big\{\,N_{K-|\mathcal{S}|}(0,\;\widehat H_{\mathcal{T}\mathcal{T}}) \in \mathbb{R}_+^{K-|\mathcal{S}|}\,\big\}, \] 
where $\mathcal{T}=\mathcal{S}^{\complement}$ denotes the complement of $\mathcal{S}$ in $\{1,\ldots,K\}$, $\widehat G = D_{\!H}^{-1/2}H\,D_{\!H}^{-1/2}$ is the polar Gram matrix, $\widehat H=\widehat G^{-1}$, and $D_{\!H}=\mathrm{diag}(H_{11},\ldots,H_{KK})$. Consequently, the weights can be computed as \begin{equation} w_j(\tilde C) \;=\; \sum_{\substack{\mathcal{S}\subseteq\{1,\ldots,K\}\\ |\mathcal{S}|=j}} \Pr\!\big\{N_{|\mathcal{S}|}(0,H_{\mathcal{S}\mathcal{S}})\in\mathbb{R}_+^{|\mathcal{S}|}\big\} \, \Pr\!\big\{N_{K-|\mathcal{S}|}(0,\widehat H_{\mathcal{T}\mathcal{T}})\in\mathbb{R}_+^{K-|\mathcal{S}|}\big\}. \label{eq:weights-gram} \end{equation} Eq.~\eqref{eq:weights-gram} reduces to the binomial expression of Eq.~\eqref{eq:orthant} when $A = \mathbb{I}_K$ (equivalently, when the cone $\tilde C$ is the positive orthant), so that $G=H=\widehat H=\mathbb{I}_K$ and each orthant probability equals $2^{-|\mathcal{S}|}$ or $2^{-|\mathcal{T}|}$. In practice, when the parameters are not orthogonal, one must compute the intrinsic volumes of the corresponding cone to obtain the mixture’s weights. Early studies such as \citet{Shapiro1985,Shapiro1988} and \citet{Robertson1988} discussed their analytical characterization for low-dimensional cases, while \citet{Sun1988} proposed recursive numerical integration for small systems and \citet{Wolak1987} provided exact expressions for a few low-dimensional cones arising in econometric inequality tests. Later, \citet{Silvapulle1996} introduced a Monte Carlo projection algorithm applicable to general polyhedral cones, and \citet{LinLindsay1997} related the weights to Weyl's tube formula, thereby linking them directly to the cone's geometric structure. More recent computational advances include the recursive integration method of \citet{Miwa2003} for moderate dimensions, the quasi-Monte Carlo scheme of \citet{GenzBretz2009} for higher-dimensional Gaussian cones, and the simulation-based framework of \citet{AmelunxenLotz2017}, which approximate the intrinsic volumes through repeated Gaussian projections. In the next sections, we investigate how these weights change in the case in which one or more coordinates are no longer treated as PoIs, but are instead ``demoted'' to nuisance parameters, so that the null hypothesis does not fix their value at zero anymore.

\section{Demoting one parameter to nuisance}
\label{sec:demotion-uncorrelated}

\subsection{Independent parameters}

Consider now the situation in which $K$ orthogonal parameters are constrained to be nonnegative, but only the first $K-1$ are of inferential interest while the last one plays the role of a nuisance parameter.
In this case, the null hypothesis cone does not correspond anymore to a single point (the origin), but to the one-dimensional ray $\tilde C_0 = \{(0,\ldots,0,t): t\ge 0\}$ lying on the boundary of the alternative cone $\tilde C = \mathbb{R}_+^K$, i.e. $\tilde C_0 = \mathbb{R}_+\, e_K$, where $e_K = (0,\ldots,0,1)^\top \in \mathbb{R}^K$ denotes the $K$-th canonical basis vector. This modification changes the geometry of the null cone, thereby altering both its intrinsic volumes and the associated projections, and hence the mixture weights. The change in weights induced by the demotion of one parameter in the orthogonal case is captured by the following lemma.

\begin{lemma}
\label{lem:delta-weights}
Let $\dot w_j^\perp=2^{-K}\binom{K}{j}$ for $j=0, \ldots,K$ denote the $\bar{\chi}^2$ weights for $K$ independent PoIs under the point null cone $\tilde C_0 = \{0\}$, and let $\bar w_j^\perp$ for $j=0, \ldots,K$ denote the corresponding weights when the $K$th parameter is treated as nuisance, i.e.\ under the ray null cone
$\tilde C_0=\mathbb{R}_+e_K$. Define $\Delta_j^\perp = \bar w_j^\perp - \dot w_j^\perp$ for $j=0,\ldots,K$.
Then
\begin{align}
&\Delta^\perp_0 = \frac{1}{2^K}, \nonumber \\[4pt]
&\Delta^\perp_j = 2^{-K}
   \Bigl[\binom{K-1}{j} - \binom{K-1}{j-1}\Bigr],
   \qquad 1 \le j \le K-1, \label{eq:deltaj}\\[4pt]
&\Delta^\perp_K =  -\frac{1}{2^K}. \nonumber
\end{align}
Equivalently, $\bar w^\perp_j = \dot w^\perp_j + \Delta^\perp_j$ for $j=0, \ldots,K$ defines a valid sequence of $\bar{\chi}^2$ weights satisfying
\[
\sum_{j=0}^{K}\Delta^\perp_j = 0,
\qquad
\sum_{j=0}^{K}\bar w_j^\perp = 1.
\]
\end{lemma}

\noindent The proof of this lemma is presented in Appendix 1.

\medskip

\noindent Lemma \ref{lem:delta-weights} shows that ``demoting'' one parameter in the orthogonal case reallocates probability mass symmetrically across adjacent components: for $K=2$, $1/4$ of the total mass is shifted from the highest-degree component ($j=2$) to the null component ($j=0$);
for $K=3$, $1/8$ is transferred simultaneously from $j=3$ to $j=1$ and
from $j=2$ to $j=0$; and analogous redistribution patterns follow for larger~$K$.

\subsection{Generalization to correlated parameters}
\label{sec:correlated}

In the orthogonal setting, replacing the point-null cone with a ray-null shifts the $\bar\chi^2$ weights by the explicit values in Eq.~\eqref{eq:deltaj}, denoted $\Delta_j^\perp$ henceforth, arising purely from this geometric change.

We now proceed to consider the general case with correlated parameters, where the linear transformation to isotropic Gaussian coordinates does alter the metric geometry of the cone, but the combinatorial modification induced by replacing the point-null with a ray-null remains the same. Because the intrinsic-volume weights vary smoothly with the Gram matrix of the transformed generators, the resulting difference between the point-null and ray-null weights remains close to its orthogonal counterpart whenever the covariance distortion is mild.
The following result shows that this ``orthogonal'' redistribution of mass is in fact stable across all covariance matrices in a compact spectral class, with deviations governed by the anisotropy of the transformation and by the number of PoIs on the boundary.

\begin{theorem}[Perturbation bounds for $\bar\chi^2$ weights under positive covariance distortion]
\label{thm:approx-delta-invariance}

Let $\tilde C\subset\mathbb{R}^K$ be the image of the alternative polyhedral cone
under the linear transformation that maps 
$Z\sim N_K(0,\Sigma)$ to an isotropic Gaussian 
$\tilde Z\sim N_K(0,\mathbb{I}_K)$, with $\Sigma = I^{-1}(\theta_0)$ a covariance matrix with positive entries.
Let $\tilde C_0^{\mathrm{pt}}=\{0\}$ denote the point-null cone corresponding to $K$ constrained PoIs, and let 
$\tilde C_0^{\mathrm{ray}}$ denote the ray-null cone obtained by demoting one of the parameters to a nuisance parameter, corresponding to the generator of $\tilde C$ associated with that parameter.

For $j=0,\ldots,K$, define
\[
\dot w_j:=w_j(\tilde C,\tilde C_0^{\mathrm{pt}};\Sigma),
\qquad
\bar w_j:=w_j(\tilde C,\tilde C_0^{\mathrm{ray}};\Sigma),
\]
and write
\[
\dot w_j^\perp:=w_j(\tilde C,\tilde C_0^{\mathrm{pt}};\mathbb{I}_K),
\qquad
\bar w_j^\perp:=w_j(\tilde C,\tilde C_0^{\mathrm{ray}};\mathbb{I}_K).
\]

\noindent In the orthogonal case ($\Sigma=\mathbb{I}_K$), the explicit differences
\(
\Delta_j^\perp
:=\bar w_j^\perp - \dot w_j^\perp
\)
are given in Lemma~\ref{lem:delta-weights}. 

\medskip

Let \(\mathcal{U}_\kappa\) denote the spectrally bounded covariance class
\[
\mathcal{U}_\kappa
:=
\bigl\{
\Sigma \succ 0:\ 
\kappa^{-1}\,\mathbb{I}_K \ \preceq\ \Sigma\ \preceq\ \kappa\,\mathbb{I}_K
\bigr\},
\]
where \(\kappa \ge 1\) is a spectral bound parameter controlling the amount of anisotropy.

For \(\Sigma \in \mathcal{U}_\kappa\), let \(G(\Sigma)\) denote the Gram matrix of the cone generators after applying the linear mapping that defines $\tilde C$, and 
define the anisotropy index
\[
\delta(\Sigma)
:=
\bigl\|\, G(\Sigma) - G(\mathbb{I}_K) \,\bigr\|_{\mathrm{op}},
\]

\noindent where \(\|\cdot\|_{\mathrm{op}}\) denotes the spectral (operator) norm.

\medskip

Then for every $\Sigma\in\mathcal{U}_\kappa$ and every $j=0,\ldots,K$,
\begin{equation}
\label{eq:delta-expansion-final}
\bar w_j
=
\dot w_j
+
\Delta_j^\perp
+
\varepsilon_j(\Sigma,K),
\end{equation}
where
\[
\varepsilon_j(\Sigma,K)
:=
\bigl(
\bar w_j-\bar w_j^\perp
\bigr)
-
\bigl(
\dot w_j - \dot w_j^\perp
\bigr),
\]
and the error terms satisfy the uniform bound
\[
\max_{0\le j\le K}|\varepsilon_j(\Sigma,K)|
\ \le\
c_K(\tilde C,\kappa)\,\delta(\Sigma),
\]
for a constant $c_K(\tilde C,\kappa)>0$ depending only on the spectral radius bound $\kappa$ and on the geometry of $\tilde C$, which in turns depends on $K$.
\end{theorem}

\bigskip

\begin{proof}
By definition of $\varepsilon_j$, in order to prove the theorem it suffices to obtain uniform bounds for the perturbations 
$|\dot w_j-\dot w_j^\perp|$ and $|\bar w_j-\bar w_j^\perp|$ over $\Sigma$. The strategy is to: (i) express each weight as a sum of Gaussian orthant probabilities; (ii) show that these probabilities are Lipschitz in the underlying covariance;  (iii) control how the relevant covariance blocks 
change with~$\Sigma$.

\medskip
From the intrinsic-volume formula recalled earlier, one has
\begin{equation}
\label{eq:rep-pt}
\dot w_j
=
\sum_{\substack{\mathcal S\subseteq\{1,\ldots,K\}\\ |\mathcal S|=j}}
P_{\mathcal S}(\Sigma)\,\widehat P_{\mathcal T}(\Sigma),
\qquad \mathcal T=\mathcal S^{\complement},
\end{equation}
where $P_{\mathcal S}(\Sigma)$ and $\widehat P_{\mathcal T}(\Sigma)$ are Gaussian 
orthant probabilities involving principal blocks of $H(\Sigma)$ and its polar.
Thus, controlling $\dot w_j$ amounts to controlling how these 
orthant probabilities vary with~$\Sigma$.

\medskip
For a covariance matrix $\Gamma$ in arbitrary dimension $d$, let
\[
\Phi_d(\Gamma)
:=\Pr\{N_d(0,\Gamma)\in\mathbb{R}_+^d\}
=\int_{\mathbb{R}_+^d}\varphi_\Gamma(x)\,dx,
\]
denote the Gaussian orthant probability associated with covariance $\Gamma$. To understand how $\Phi_d$ changes under perturbations of $\Gamma$, consider the symmetric perturbation $\Gamma_t=\Gamma+t\Delta\Gamma$.  
Standard matrix calculus identities yield
\[
\partial_{\Delta\Gamma}\varphi_\Gamma(x)
=\frac12\,\varphi_\Gamma(x)\Bigl(
x^\top\Gamma^{-1}(\Delta\Gamma)\Gamma^{-1}x
-\operatorname{tr}(\Gamma^{-1}\Delta\Gamma)\Bigr).
\]

For $\Gamma_1,\Gamma_2 \in \mathcal K_{d,M}$, the segment 
\(
\Gamma_t = (1-t)\Gamma_1 + t\Gamma_2 ,
\mbox{ with } t\in[0,1],
\)
remains in the compact spectral class $\mathcal K_{d,M}$, which is convex.  
Hence the derivative bound we derive below will hold uniformly in  $t\in[0,1]$.  
Moreover, since $\Gamma_t$ stays in $\mathcal K_{d,M}$, there exist $c,C>0$ such that 
\[
\bigl|\partial_{\Delta\Gamma}\varphi_{\Gamma_t}(x)\bigr|
\le C e^{-c\|x\|^2}
\quad\text{for all }t\in[0,1],\ x\in\mathbb R^d,
\]
and the right-hand side is integrable on $\mathbb{R}_+^d$.  
Thus the derivative may be passed under the integral sign by dominated convergence,
\[
D\Phi_d(\Gamma)[\Delta\Gamma]
=
\int_{\mathbb{R}_+^d}\partial_{\Delta\Gamma}\varphi_\Gamma(x)\,dx.
\]

Assume 
\[
\Gamma\in\mathcal{K}_{d,M}
:=\{\Gamma\succ0:\ M^{-1}\mathbb{I}_d\preceq\Gamma\preceq M\mathbb{I}_d\}.
\]
Then $\|\Gamma^{-1}\|\le M$, and therefore
\begin{equation}
|x^\top\Gamma^{-1}(\Delta\Gamma)\Gamma^{-1}x|
\le M^2\|\Delta\Gamma\|\,\|x\|^2,\qquad
|\operatorname{tr}(\Gamma^{-1}\Delta\Gamma)|
\le dM\,\|\Delta\Gamma\|.
\label{bounds}
\end{equation}

Using \eqref{bounds} in the derivative formula yields
\[
|\partial_{\Delta\Gamma}\varphi_\Gamma(x)|
\le \tfrac12\,\varphi_\Gamma(x)
\bigl(M^2\|\Delta\Gamma\|\,\|x\|^2 + dM\|\Delta\Gamma\|\bigr).
\]

Integrating and bounding the truncated moments of $N_d(0,\Gamma)$ gives
\begin{equation}
\label{eq:DPhid-bound-proof}
|D\Phi_d(\Gamma)[\Delta\Gamma]|
\le c(d,M)\,\|\Delta\Gamma\|,\qquad
c(d,M)=\tfrac12 dM(M^2+1).
\end{equation}

Therefore, Gaussian orthant probabilities are uniformly Lipschitz in the  covariance matrix on compact spectral classes.

\medskip

The derivative bound \eqref{eq:DPhid-bound-proof} will now be used to obtain a global Lipschitz bound for $\Phi_d$ on the spectral class $\mathcal K_{d,M}$. In particular, we show that $\Phi_d$ is Lipschitz in $\Gamma$, meaning that small perturbations of the covariance produce proportionally small changes in the orthant probability.

For $\Gamma_1,\Gamma_2\in\mathcal{K}_{d,M}$ define
$\Gamma_t=(1-t)\Gamma_1+t\Gamma_2$. Since the spectral class is convex and \eqref{eq:DPhid-bound-proof} is uniform,
\[
\Phi_d(\Gamma_2)-\Phi_d(\Gamma_1)
=\int_0^1 D\Phi_d(\Gamma_t)[\Gamma_2-\Gamma_1]\,dt.
\]
Hence
\begin{equation}
\label{eq:PhiLip-proof}
|\Phi_d(\Gamma_2)-\Phi_d(\Gamma_1)|
\le c(d,M)\,\|\Gamma_2-\Gamma_1\|.
\end{equation}
that is, $\Phi_d$ is Lipschitz in $\Gamma$ on the spectral class 
$\mathcal{K}_{d,M}$.

\medskip

The map $\Sigma\mapsto G(\Sigma)$ is smooth on the compact set 
$\mathcal U_\kappa$, and inversion is smooth on the positive definite cone, hence
\begin{equation}
\label{eq:inv-perturb-proof}
\|H(\Sigma)-H(\mathbb{I}_K)\|
\le c_1(\tilde C,\kappa)\,\delta(\Sigma).
\end{equation}
Principal submatrices satisfy the same bound:
\begin{equation}
\label{eq:block-perturb-proof}
\|H_{\mathcal S\mathcal S}(\Sigma)-H_{\mathcal S\mathcal S}(\mathbb{I}_K)\|
\le c_1(\tilde C,\kappa)\,\delta(\Sigma),
\end{equation}
and similarly for $\widehat H$, $H^{\mathrm{ray}}$, and $\widehat H^{\mathrm{ray}}$.

\medskip

From \eqref{eq:rep-pt},
\[
|\dot w_j-\dot w_j^\perp|
\le
\sum_{|\mathcal S|=j}
\Bigl(
|P_{\mathcal S}(\Sigma)-P_{\mathcal S}^\perp|
+
|\widehat P_{\mathcal T}(\Sigma)-\widehat P_{\mathcal T}^\perp|
\Bigr),
\]
where $\mathcal T=\mathcal S^{\complement}$.
Applying the Lipschitz bound \eqref{eq:PhiLip-proof} together with the block 
perturbation bound \eqref{eq:block-perturb-proof}, we obtain for each fixed
subset $\mathcal S$,
\[
|P_{\mathcal S}(\Sigma)-P_{\mathcal S}^\perp|
\le c(j,\kappa)\,\delta(\Sigma),
\qquad
|\widehat P_{\mathcal T}(\Sigma)-\widehat P_{\mathcal T}^\perp|
\le c(K-j,\kappa)\,\delta(\Sigma),
\]
for suitable constants $c(j,\kappa)$ depending only on the dimension of the
block and on~$\kappa$.  Since, for fixed $K$, there are only finitely many
subsets $\mathcal S\subseteq\{1,\ldots,K\}$ of size $j$, we may absorb the
resulting finite sum into a single constant depending on $K$ and on the
geometry of $\tilde C$, and conclude that
\begin{equation}
\label{eq:dot-final-proof}
|\dot w_j-\dot w_j^\perp|
\le c_K(\tilde C,\kappa)\,\delta(\Sigma),
\end{equation}
for some $c_K(\tilde C,\kappa)>0$.

\medskip

The same argument applies verbatim to the covariance blocks 
$H^{\mathrm{ray}}(\Sigma)$ and $\widehat H^{\mathrm{ray}}(\Sigma)$,
yielding
\begin{equation}
\label{eq:bar-final-proof}
|\bar w_j-\bar w_j^\perp|
\le c_K(\tilde C,\kappa)\,\delta(\Sigma).
\end{equation}

Combining \eqref{eq:dot-final-proof} and \eqref{eq:bar-final-proof} and using the 
triangle inequality explicitly,
\[
|\varepsilon_j(\Sigma,K)|
=
\bigl|
(\bar w_j(\Sigma)-\bar w_j^\perp)
-
(\dot w_j(\Sigma)-\dot w_j^\perp)
\bigr|
\le
|\bar w_j(\Sigma)-\bar w_j^\perp|
+
|\dot w_j(\Sigma)-\dot w_j^\perp|
\le
2\,c_K(\tilde C,\kappa)\,\delta(\Sigma),
\]
and the factor of $2$ can be absorbed into the constant $c_K(\tilde C,\kappa)$.
Substituting into \eqref{eq:delta-expansion-final} completes the proof.

\end{proof}

\begin{remark}
\label{rem:approx-interpretation}

Theorem~\ref{thm:approx-delta-invariance} shows that the difference between the
$\bar{\chi}^2$ weights associated with the point-null and ray-null cones,
\[
\Delta_j(\Sigma)
=\bar w_j(\Sigma)-\dot w_j(\Sigma),
\qquad j=0,\ldots,K,
\]
deviates from its orthogonal counterpart $\Delta_j^{\perp}$ by an amount
bounded by a term that depends on $K$ only through the geometry of $\tilde C$
(e.g.\ its face lattice), and on $\Sigma$ solely through the anisotropy index
$\delta(\Sigma)$.
This has several implications.

\smallskip

\emph{(i) Orthogonal $\Delta$ is the leading-order term.}
In the orthogonal case $\Sigma=\mathbb{I}_K$, the vector $\Delta^\perp$ has the exact closed-form expression provided in Lemma \ref{lem:delta-weights}. The theorem shows that $\Delta^\perp$ remains the dominant
contribution for general $\Sigma$, and correlations introduce only a controlled perturbation of size at most $c_K(\tilde C,\kappa)\,\delta(\Sigma)$.

\smallskip

\emph{(ii) Growth in $K$.}
The theorem does not impose a specific growth rate in $K$, but it isolates the $K$-dependence entirely inside $c_K(\tilde C,\kappa)$ and shows that for a fixed geometry the error varies smoothly with $\delta(\Sigma)$.

\smallskip

\emph{(iii) Dependence on the anisotropy of $\Sigma$.}
The quantity $\delta(\Sigma)=\|G(\Sigma)-G(\mathbb{I}_K)\|_{\mathrm{op}}$
measures the deviation of the transformed cone from the orthogonal case. Therefore, the theorem implicitely states that the effect on the $\bar\chi^2$ weights of demoting a generator to a nuisance direction depends primarily on the local geometry of the cone as encoded by $G(\Sigma)$, and varies in a Lipschitz fashion with the geometric deviation from the orthogonal case. The approximation is excellent when $\Sigma$ is nearly diagonal and degrades smoothly as the covariance becomes more anisotropic.

\smallskip

\emph{(iv) Practical implication.}
The theorem supports the use of the ``orthogonal-difference'' approximation in applications: starting from the point-null weights and adding the orthogonal $\Delta^\perp$ yields an accurate approximation of the true ray-null weights whenever the covariance distortion is not extreme.

\end{remark}

\begin{remark}
\label{rem:positive-corr-necessity-concise}

Strictly speaking, the proof of Theorem~\ref{thm:approx-delta-invariance} does not require the entries of~$\Sigma$ to be positive. However, positive correlation is needed for the asymptotic distribution of $\lrs$ to admit a genuine $\bar\chi^2$-mixture interpretation only under non-negative correlation \citep{us}, as negative correlation alters the cone geometry in such a way that a valid $\bar\chi^2$ representation may fail. Thus, while the perturbation expansion~\eqref{eq:delta-expansion-final} holds for all $\Sigma\in\mathcal U_\kappa$, its probabilistic interpretation as a perturbation of $\bar\chi^2$ weights implicitly relies on the positive-correlation regime. For $\rho<0$, the extension proposed in the two-parameter case in \citet{us} can in principle be extended to $K$ parameters, but such an extension lies beyond the scope of this work.

\end{remark}

\begin{corollary}[Equicorrelation]
\label{cor:equicorr-bound}
Let 
\[
\Sigma_\rho = (1-\rho)\,\mathbb{I}_K + \rho\,\mathbf{1}\mathbf{1}^\top,
\qquad -\frac{1}{K-1} < \rho < 1,
\]
be an equicorrelation covariance matrix, with condition number
\[
\kappa(\Sigma_\rho)
=
\frac{1 + (K-1)\rho}{1-\rho}
\quad\text{so that}\quad
\kappa(\Sigma_\rho)-1 = \frac{K\rho}{1-\rho}.
\]
Then its anisotropy index satisfies
\[
\delta(\Sigma_\rho)
=\|G(\Sigma_\rho)-G(\mathbb{I}_K)\|_{\mathrm{op}}
\;\le\;
\frac{K\rho}{1-\rho}.
\]
Consequently, Theorem~\ref{thm:approx-delta-invariance} yields
\[
\max_{0\le j\le K} |\varepsilon_j(\Sigma_\rho,K)|
\;\le\;
c_K(\tilde C,\kappa(\Sigma_\rho))\,
\frac{K\rho}{1-\rho}.
\]
In particular, for fixed $\rho<1$, the deviation from the orthogonal
difference $\Delta^\perp$ grows at most on the order of $K\rho/(1-\rho)$,
up to the geometric constant $c_K(\tilde C,\kappa(\Sigma_\rho))$.
\end{corollary}

\begin{proof}
The eigenvalues of $\Sigma_\rho$ are 
$\lambda_1 = 1+(K-1)\rho$ and $\lambda_2=\dots=\lambda_K=1-\rho$,
giving the stated condition number.  
The inverse is
\[
\Sigma_\rho^{-1}
=
\frac{1}{1-\rho}\,I
-\frac{\rho}{(1-\rho)\{1+(K-1)\rho\}}\,
\mathbf{1}\mathbf{1}^\top,
\]
so all diagonal entries are equal to
\[
a=\frac{1+(K-2)\rho}{(1-\rho)\{1+(K-1)\rho\}},
\]
and all off-diagonal entries equal
\[
b=-\,\frac{\rho}{(1-\rho)\{1+(K-1)\rho\}}.
\]
Since $G(\Sigma_\rho)=D^{-1/2}\Sigma_\rho^{-1}D^{-1/2}$ with 
$D=a\,\mathbb{I}_K$, we obtain
\[
G(\Sigma_\rho)
=\mathbb{I}_K - \alpha\bigl(\mathbf{1}\mathbf{1}^\top-\mathbb{I}_K\bigr),
\qquad
\alpha=\frac{\rho}{1+(K-2)\rho}.
\]
Thus $G(\Sigma_\rho)-\mathbb{I}_K$ has eigenvalues $\alpha$ (multiplicity $K-1$)
and $-\alpha(K-1)$, so
\[
\delta(\Sigma_\rho)
=\|G(\Sigma_\rho)-\mathbb{I}_K\|_{\mathrm{op}}
=\alpha(K-1)
=\frac{(K-1)\rho}{1+(K-2)\rho}
\;\le\;
\frac{K\rho}{1-\rho}.
\]
Substituting this bound into Theorem~\ref{thm:approx-delta-invariance}
gives the claim.
\end{proof}

\medskip

\noindent Corollary~\ref{cor:equicorr-bound} shows that for equicorrelated covariance matrices $\Sigma_\rho$, 
the bound is explicitly $\delta(\Sigma_\rho)\le K\rho/(1-\rho)$, guaranteeing the perturbation to remain small for moderate $K$ and $\rho$.

\section{Generic number of nuisance parameters}

\subsection{Independent parameters}

\begin{lemma}
\label{lem:delta-weights-m}
Let $\dot w^\perp_j = 2^{-K}\binom{K}{j}$ for $j=0,\ldots,K$ denote the $\bar{\chi}^2$ weights for $K$ independent PoIs constrained to be nonnegative, with point-null cone $C^{\mathrm{pt}}_0=\{0\}$. For a fixed integer $m$ with $1 \le m \le K-1$, 
let $\bar w^{(m)\perp}_j$, $j=0,\ldots,K$, denote the $\bar{\chi}^2$ weights when the last $m$ coordinates are treated as nuisance parameters on the boundary.

In the case where the last $m$ parameters are constrained nuisance
parameters (i.e.\ $\theta_{K-m+1},\ldots,\theta_K \ge 0$ under $H_0$),
the corresponding null cone is the $m$-dimensional face
\[
\tilde C^{(m)}_0
= \bigl\{ (0,\ldots,0, t_{K-m+1},\ldots,t_K) : t_j \ge 0 \, \mbox{ for } \, j=K-m+1, \ldots, K \bigr\}
\]
or, equivalently,
\(
\tilde C^{(m)}_0
= \operatorname{cone}(e_{K-m+1},\ldots,e_K)
\cong \mathbb{R}_+^m.
\) \\
Then, 

\[
\Delta^{(m)\perp}_j = \bar w^{(m)\perp}_j - \dot w^{\perp}_j
  = \begin{cases}
      2^{-K}\!\left[\,2^m \binom{K-m}{j} - \binom{K}{j}\right], & 0 \le j \le K-m,
      \\[6pt]
      -\,2^{-K} \binom{K}{j}, & K-m+1 \le j \le K.
    \end{cases}
\]
and
$\bar w^{(m)\perp}_j = \dot w^{\perp}_j + \Delta^{(m)\perp}_j$ for $j=0,\ldots,K$, with
\[
\sum_{j=0}^{K} \Delta^{(m)\perp}_j = 0,
\qquad
\sum_{j=0}^{K} \bar w^{(m)\perp}_j = 1.
\]

Equivalently,
\[
\bar w^{(m)\perp}_j
  = \begin{cases}
      2^{-(K-m)} \binom{K-m}{j}, & 0 \le j \le K-m,\\[6pt]
      0, & K-m+1 \le j \le K,
    \end{cases}
\]

and for $m=1$ this reduces to Lemma~\ref{lem:delta-weights}.
\end{lemma}

\noindent The proof of this lemma is presented in Appendix 1.

\subsection{Correlated parameters}

Theorem~\ref{thm:approx-delta-invariance} substantially legitimates, from a theoretical point of view, the usage of the $\bar \chi$ weights of the point-null problem plus the orthogonal difference pattern $\Delta^{\perp}$ when demoting one parameter to nuisance.  The same argument extends, in principle, to a generic number $m\ge1$ of boundary nuisance parameters: the only change in the proof would be that the ray-null cone is replaced by an $m$-dimensional face of $\tilde C$, the perturbation bounds remain valid.  
Formally, one obtains an error bound of the form
\[
\max_{j} |\varepsilon^{(m)}_j(\Sigma;K)|
\;\le\;
c_{K,m}(\tilde C)\, \delta(\Sigma),
\]
which reduces to the result of Theorem~\ref{thm:approx-delta-invariance} when $m=1$.

However, the already loose bound of Theorem~\ref{thm:approx-delta-invariance} deteriorates further when $m > 1$, becoming too coarse to be informative. As a consequence, the approximation based on the orthogonal difference pattern is no longer sufficiently accurate to offer practical benefit.

For this reason, for the case $m > 1$ with correlated parameters we introduce a different approximation strategy that
does not rely on the orthogonal difference pattern.  The idea is heuristic but grounded in geometric considerations: each
face $F_I$ of the alternative cone contributes an ``effective'' number of degrees of freedom that depends on how the PoIs and
the nuisance coordinates interact through the Fisher information matrix.
This leads to a rank-based reaggregation of the face masses $m_S$, which
in practice yields a much sharper approximation to the LAN limit
distribution of the likelihood ratio. The construction is detailed below.

\begin{definition}[Rank-based approximate $\bar{\chi}^2$ weights with $m$ nuisance parameters]
Let the parameter be partitioned as
\(
\vartheta=(\psi,\gamma), 
\) with \( 
\psi\in\mathbb{R}^p,\ \gamma\in\mathbb{R}^m,
\)
and consider the one-sided hypothesis
\(
H_0:\psi=0,\ \gamma\ge 0
\, \, \text{vs} \, \,
H_1:\psi\ge 0,\ \gamma\ge 0.
\)
Let $K=p+m$, let $I(\vartheta_0)$ denote the Fisher information matrix at 
$\vartheta_0=0$, and write $\mathcal{P}$ and $\mathcal{N}$ for the index sets of the $p$ parameters of interest and the $m$ nuisance parameters on the boundary, respectively.

As in Section~\ref{sec:background}, let 
\[
\tilde Z = D^{1/2}P^{\top} Z
\qquad\text{and}\qquad
\tilde C = D^{1/2}P^{\top}\mathbb{R}_+^{K}.
\]
denote the transformed isotropic Gaussian score and the corresponding transformed alternative cone.  Let $\{\tilde F_S : S\subseteq\{1,\ldots,K\}\}$ be the faces of $\tilde C$,
with intrinsic-volume mass $m_S=\alpha(\tilde F_S)\,\beta(\tilde F_S)$.
For any face $\tilde F_S$, the index set $S$ lists the coordinates whose
inequality constraints are active, meaning that the projection onto
$\tilde F_S$ sets precisely the coordinates in $S$ to strictly positive values, while those outside $S$ bind at zero.

When all $K$ coordinates are treated as PoIs on the boundary, the point-null $\bar{\chi}^2$ weights are obtained by grouping face masses according to $|S\cap \mathcal{P}|$, the number of active PoIs:
\[
w_u^{\mathrm{(point)}} 
   =\sum_{S:\,|S\cap \mathcal{P}|=u} m_S,
   \qquad u=0,\ldots,p.
\]

For the ray-null problem $H_0:\psi=0,\ \gamma\ge 0$, the face $\tilde F_S$ may activate both PoIs and nuisance coordinates, say $S_\mathcal{P}=S\cap \mathcal{P}$ and $S_\mathcal{N}=S\cap \mathcal{N}$.  
To quantify how many PoI directions remain informative after accounting for the nuisance components active on that face, we consider the principal
Fisher information block $I(\vartheta_0)_{SS}$ written as
\[
I_{SS}=
\begin{pmatrix}
I_{\mathcal{P}\mathcal{P}} & I_{\mathcal{P}\mathcal{N}} \\
I_{\mathcal{NP}} & I_{\mathcal{NN}}
\end{pmatrix},
\]
and define the face-wise effective Fisher information via the Schur complement
\[
I_{\mathrm{face}}
   = I_{\mathcal{PP}} - I_{\mathcal{PN}} I_{\mathcal{NN}}^{-1} I_{\mathcal{NP}}.
\]
The corresponding \emph{effective rank}
\[
r(S)=\operatorname{rank}(I_{\mathrm{face}})
   \in\{0,1,\ldots,|S_\mathcal{P}|\}
\]
represents the number of PoI directions that remain linearly independent of the nuisance directions on face $S$.

The rank-based approximate $\bar{\chi}^2$ weights for the ray-null problem are obtained by grouping face masses according to this effective rank:
\[
w_u^{\mathrm{(ray)}}
   =\sum_{S:\,r(S)=u} m_S,
   \qquad u=0,\ldots,p,
\]
followed by normalization so that $\sum_{u=0}^p w_u^{\mathrm{(ray)}}=1$.
These $w_u^{\mathrm{(ray)}}$ are referred to as the \emph{rank-based approximate
$\bar{\chi}^2$ weights} for testing $p$ nonnegative PoIs in the
presence of $m$ correlated nuisance parameters on the boundary.
\end{definition}

\begin{remark}
The rank-based approximation groups together all faces of the transformed cone
$\tilde C$ that contribute the same effective number of degrees of freedom,
computed as the rank of the Schur complement associated with that face, and
sums their intrinsic-volume masses to form the approximate $\bar{\chi}^2$ weights.
This is an approximation in the presence of correlated parameters, but it
becomes an exact approach when the Fisher information matrix is diagonal, i.e.\
when all $K=p+m$ coordinates are orthogonal.

Indeed, if $I$ is diagonal then $I_{\mathcal{PN}}=I_{\mathcal{NP}}=0$ for every face $F_S$, and the
face-wise Schur complement reduces to $I_{\mathrm{face}} = I_{\mathcal{PP}}$.  Since
$I_{\mathcal{PP}}$ is a diagonal matrix of size $|S_\mathcal{P}|\times |S_\mathcal{P}|$ with strictly positive
entries on the active PoI coordinates, we have
\[
r(S) = \operatorname{rank}(I_{\mathrm{face}}) = |S_\mathcal{P}|.
\]
Therefore the rank-based aggregation rule
\[
w_u^{(\mathrm{ray})}
= \sum_{S:\, r(S)=u} m_S
\]
coincides exactly with the grouping of face masses by the number of active PoI, i.e.\ with the exact $\bar{\chi}^2$ weights obtained by
summing over all faces $S$ such that $|S_\mathcal{P}| = u$.
For orthogonal parameters the face masses satisfy $m_S = 2^{-K}$, and the number
of faces with $|S_\mathcal{P}|=u$ equals $\binom{p}{u} 2^m$.  Consequently
\[
w_u^{(\mathrm{ray})}
= 2^{-K} \binom{p}{u} 2^m
= 2^{-p}\binom{p}{u},
\qquad u=0,\dots,p,
\]
which are precisely the classical $\bar\chi^2$ weights for $p$
independent one-sided parameters.  Hence, in the independent case, the proposed
rank-based weights recover the exact result.

Deviations from the exact ray-null $\bar\chi^2$ weights occur precisely when the PoI directions on a face are strongly collinear with the nuisance directions, so that
$I_{\mathrm{face}}$ has reduced rank; in this sense, the approximation is most accurate when the PoI components retain substantial information that is not explainable by the
nuisance block.
\end{remark}

\section{Numerical validation} \label{sec:simulations}

In this section we assess the accuracy of the $\bar{\chi}^2$ weight approximations proposed in the preceding sections for the correlated case. The numerical validation of Lemma \ref{lem:delta-weights} and Lemma \ref{lem:delta-weights-m} is provided in Appendix 2. All simulations are conducted directly in the LAN regime, by drawing
samples from the Gaussian limit experiment
\(
Z \sim \mathcal{N}_K(0,\Sigma) \) with \( \Sigma = I(\theta_0)^{-1}
\).  
In each configuration we compare the empirical distribution of the
likelihood-ratio statistic obtained with $10^5$ Monte Carlo (MC) repetitions to the corresponding $\bar{\chi}^2$ mixture.

For each configuration, we provide figures comparing the cumulative distribution function (CDF) of both distributions. The agreement is quantified through the metrics displayed within each panel: a comparison between the 50\textsuperscript{th} and 95\textsuperscript{th} empirical and theoretical quantiles of $\lrs$; the Kolmogorov-Smirnov type distance
\[
D_\infty \;=\; \sup_{t\ge 0} \bigl| F_{\mathrm{emp}}(t) - F_{\mathrm{mix}}(t) \bigr|,
\]
i.e., the maximal vertical difference between the empirical cumulative distribution function \(F_{\mathrm{emp}}\) and the theoretical mixture CDF \(F_{\mathrm{mix}}\); a calibration diagnostic at nominal level \(\alpha=0.05\), defined as
\[
\text{Tail}/\alpha
\;=\;
\frac{1 - F_{\mathrm{emp}}(t_{0.95}^{\mathrm{mix}})}{\alpha},
\]
where \(t_{0.95}^{\mathrm{mix}}\) denotes the \(95\%\) quantile of the theoretical mixture.

\subsection{Results for Theorem 1}
To assess the performance of the orthogonal weight-difference approximation in correlated settings, we examine $K \in \{4,7,10\}$ using mildly and strongly correlated covariance matrices $\Sigma$, obtained by drawing correlations uniformly from [0, 0.5] and [0.5, 0.9] respectively, and transforming to valid positive-definite matrices. 

Figure~\ref{fig:theorem1} reports the results for the mild-correlation regime: we observe that the empirical CDFs of $\lrs$ remain almost indistinguishable from the analytic approximation (red) for all $K \in \{4,7,10\}$. Although slightly larger than in the orthogonal case -- for which we obtained a distance equal to 0.002 for all the three values of $K$, see Appendix 2 --, the Kolmogorov distance metrics $D_\infty$ reported in each panel remain very small, the largest value being 0.026 for $K = 10$. The upper-tail diagnostic is also close to one, the lowest value being 0.895 for $K = 10$, indicating that any distortion of the distribution, including its extreme upper tail, is quite moderate. Thus, even in a mild correlated regime, the orthogonal closed-form weights provide an effective approximation to the true $\bar{\chi}^2$ weights after demoting one parameter to nuisance. At the same time, the diagnostics follow the qualitative trend predicted by the theorem, as the discrepancies increase with $K$.

A comparison with the strong-correlation regime in Figure~\ref{fig:theorem1_bis} shows increasing discrepancies with $K$, though still mild: for each value of $K$, both the Kolmogorov distance $D_\infty$ and the upper-tail diagnostic deviate more from their ideal values than in the mildly correlated case, indicating larger perturbations of the $\bar{\chi}^2$ weights. Nevertheless, the discrepancies remain moderate: $D_\infty$ stays well below $0.1$, and the tail diagnostic departs from one by at most a few tens of percent. The empirical CDFs are still very close to the analytic approximation and the reported quantile differences remain small, so the orthogonal closed-form weights continue to provide a remarkably accurate approximation to the true $\bar{\chi}^2$ weights even under strong correlation after demoting one parameter to nuisance.

These findings might indicate that the remainder term in~\eqref{eq:delta-expansion-final} inevitably increases with the number of constrained parameters, also in light of the theoretical bound from Theorem~\ref{thm:approx-delta-invariance} and of Corollary~\ref{cor:equicorr-bound} which shows that $\delta(\Sigma_\rho)$ grows with $K$ in the equicorrelated case.
To disentangle the effect of dimension from the effect of correlation structure, we examine a controlled equicorrelation setting $\Sigma_\rho$ with fixed $\rho=0.5$ and report in Figure~\ref{fig:theorem1_ter} the Kolmogorov distance $D_\infty$ for $K=2,\ldots,10$, alongside the corresponding anisotropy index. Despite the monotone growth of $\delta(\Sigma_\rho)$ with $K$, the Kolmogorov distances remain uniformly small and exhibit no systematic increase, indicating that the larger discrepancies observed previously for $K=10$ under strong correlation are primarily a consequence of increased anisotropy rather than the dimensionality itself. This confirms that the orthogonal closed-form correction $\Delta_j^\perp$ provides a stable and accurate approximation to the true ray-null $\bar{\chi}^2$ mixture across a wide range of dimensions, and that correlation strength, not the larger $K$, is the dominant source of perturbation in practice.

\begin{figure}
    \centering
    \includegraphics[width=0.32\linewidth]{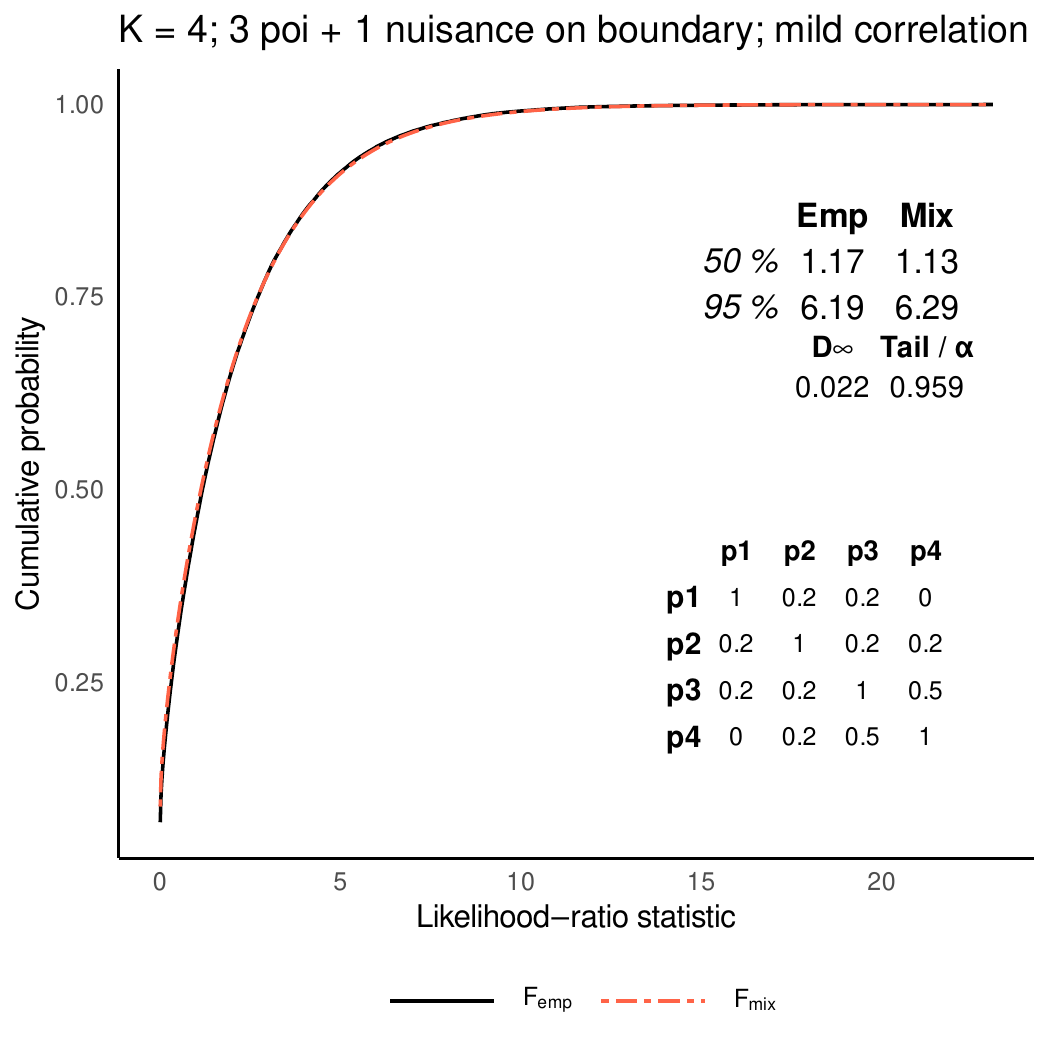}
    \includegraphics[width=0.32\linewidth]{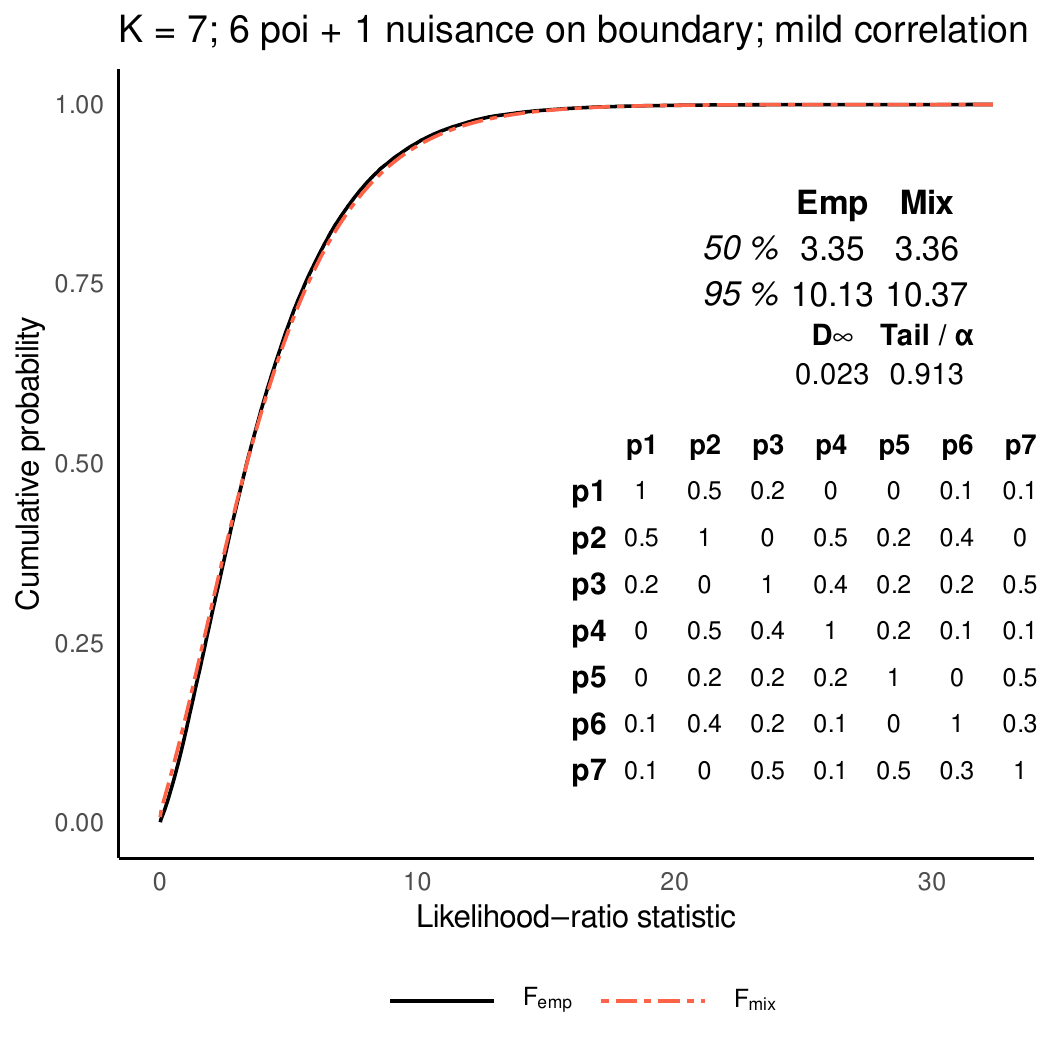}
    \includegraphics[width=0.32\linewidth]{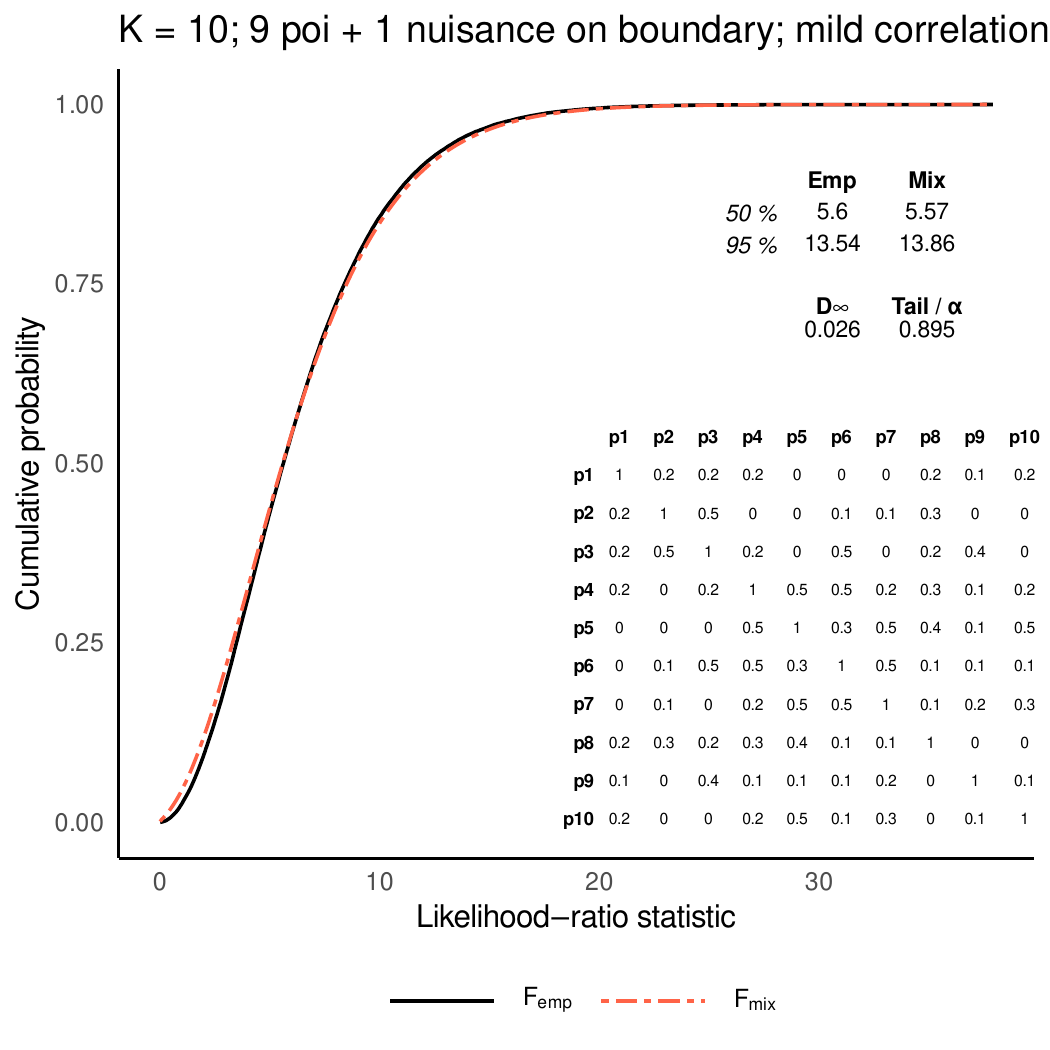}
    \caption{Empirical (solid black) versus approximated theoretical (dashed red) CDFs of $\lrs$ for $K=4,7,10$ in the mildly correlated case with one nuisance parameter on the boundary.}
    \label{fig:theorem1}
\end{figure}

\begin{figure}
    \centering
    \includegraphics[width=0.32\linewidth]{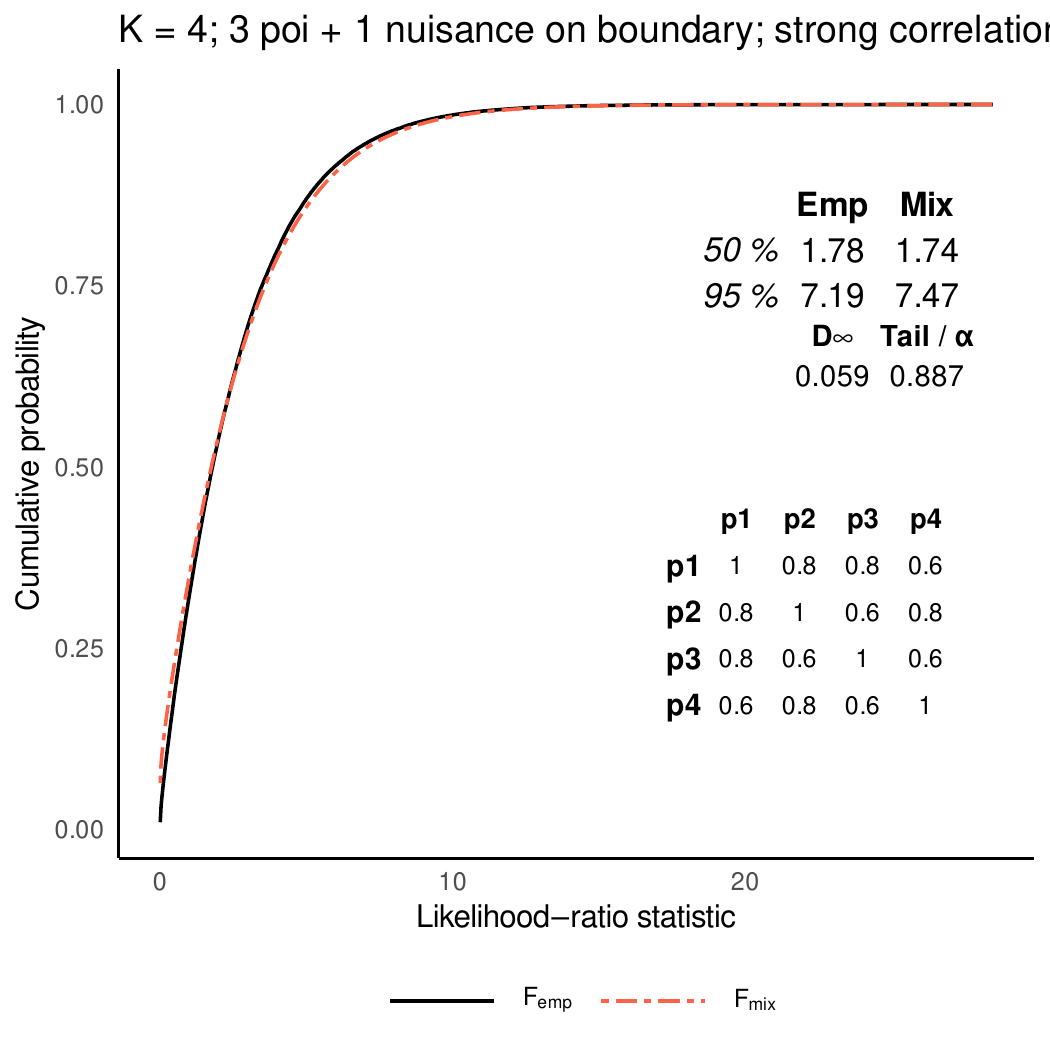}
    \includegraphics[width=0.32\linewidth]{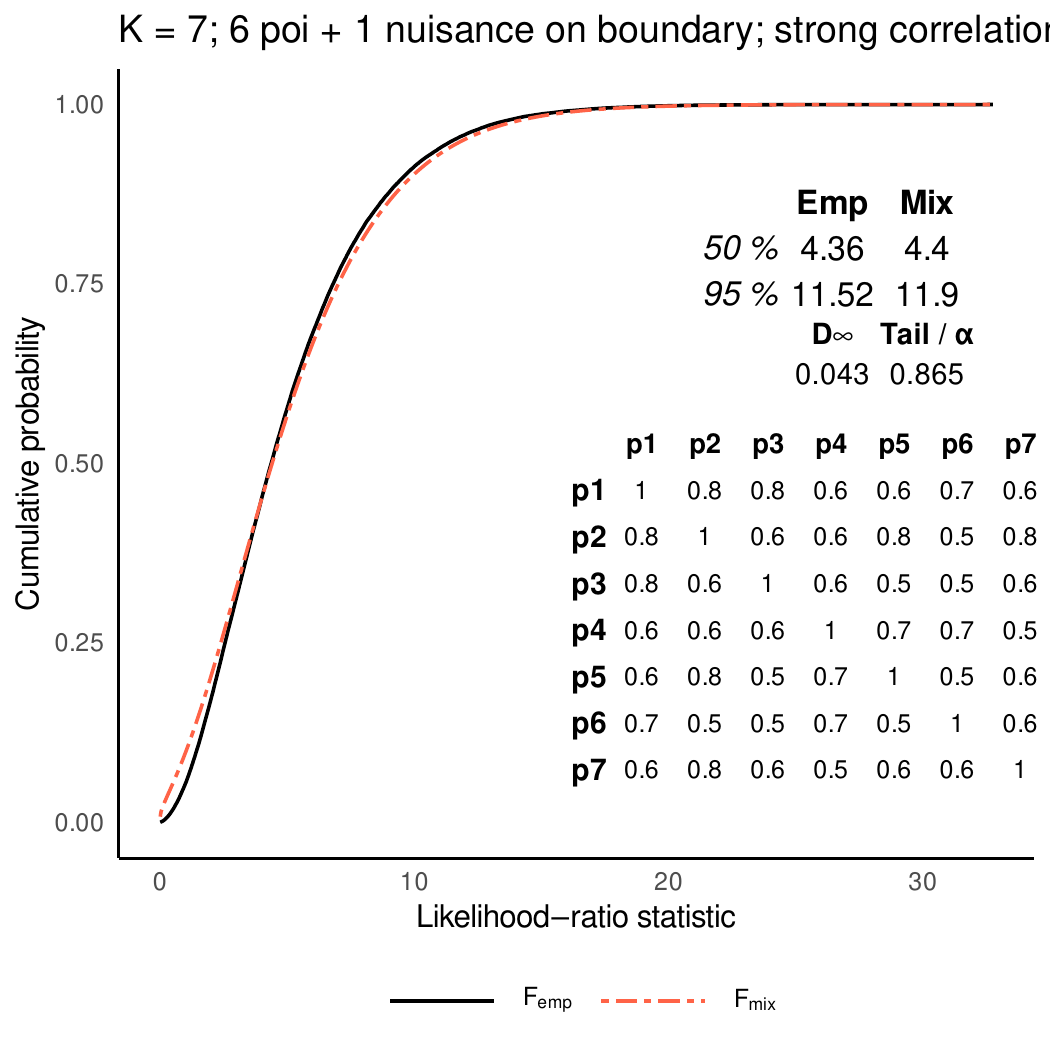}
    \includegraphics[width=0.32\linewidth]{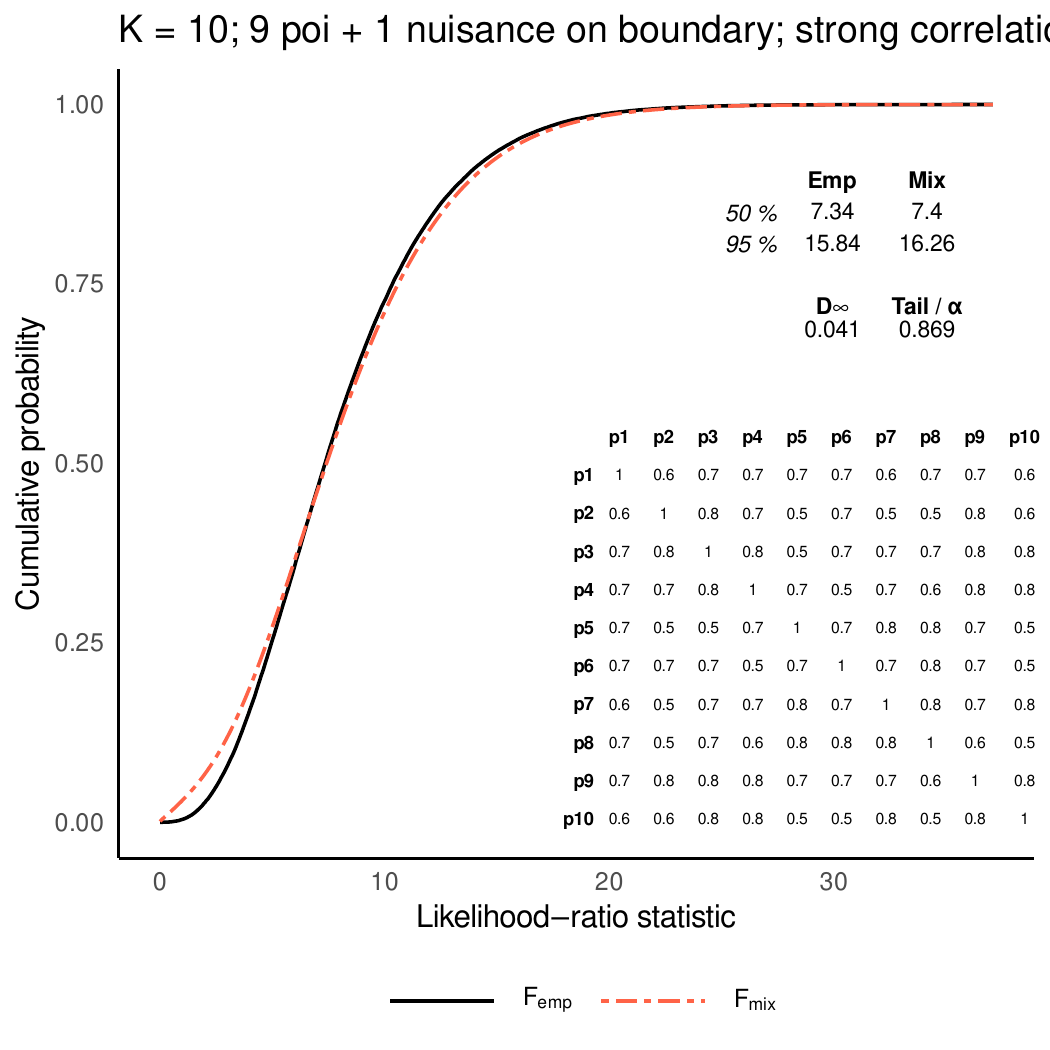}
    \caption{Empirical (solid black) versus approximated theoretical (dashed red) CDFs of the $\lrs$ for $K=4,7,10$ in the strongly correlated case with one nuisance parameter on the boundary.}
    \label{fig:theorem1_bis}
\end{figure}

\begin{figure}
    \centering
    \includegraphics[width=0.6\linewidth]{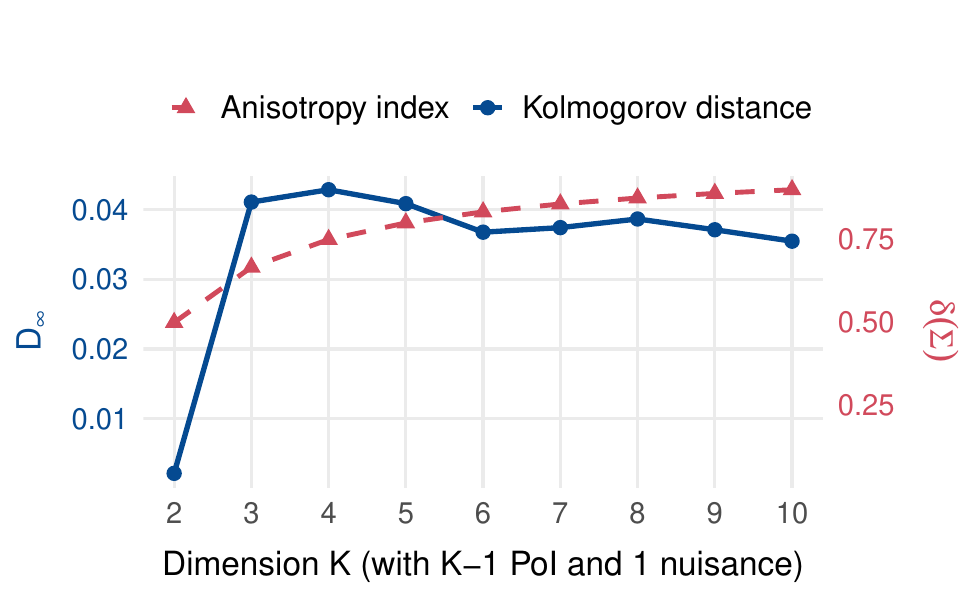}
    \caption{Kolmogorov distance $D_\infty$ (blue) and anisotropy index $\delta(\Sigma_\rho)$ (red) as functions of $K$ for $\Sigma_\rho$ with an equicorrelated covariance with $\rho=0.5$ for the case of $K-1$ PoIs and 1 nuisance parameter on the boundary.}
    \label{fig:theorem1_ter}
\end{figure}

\subsection{Results for the rank-based approximation}

We now turn to the case of multiple nuisance parameters on the boundary ($m>1$) under general correlation structures, for which we propose a heuristic approximation based on the effective degrees of freedom contributed by each face of the alternative cone—exact in the orthogonal case, and empirically adequate beyond it. Figures~\ref{fig:heur1}–\ref{fig:heur2} reveal a consistent pattern: the accuracy of the rank-based mixture is affected far more by the total dimension $K$ of the PoI vector than by correlation strength, with the latter playing a clearly secondary role.

In the mildly correlated regime, the empirical and theoretical CDFs remain very close for $K=4$ and $K=7$ ($D_\infty = 0.030$ and $0.049$, respectively, with tail ratios around $1.27$–$1.32$). For $K=10$, discrepancies become more visible: $D_\infty$ rises to $0.123$ and the $5\%$ tail ratio to $1.81$, signalling that the theoretical mixture increasingly underestimates the right tail. The approximation remains usable, but the tail miscalibration should be explicitly accounted for in inferential applications.

In the strongly correlated setting, the deterioration relative to the mild case is less dramatic than one might expect and is pronounced mainly at $K=4$, where $D_\infty$ increases from $0.03$ to $0.10$ and the tail ratio from $1.27$ to $1.90$. For $K=7$ and $K=10$, the diagnostics remain similar to the mildly correlated regime, reinforcing the conclusion that the growth of the PoI dimension -- not correlation strength -- is the dominant factor driving deviations between empirical and theoretical CDFs.

Figure~\ref{fig:heur3} further clarifies this point: for equicorrelated structures $\Sigma_{\rho}$ with fixed $\rho=0.5$, $D_\infty$ shows a mild upward trend with increasing $K$, levelling off at moderate dimensions. Here, $D_\infty$ no longer tracks the anisotropy index as closely as in the single-nuisance case (Figure~\ref{fig:theorem1_ter}), suggesting that with multiple boundary nuisances the approximation error depends more on the interplay between PoI dimensionality and cone geometry than on covariance anisotropy alone. A complementary perspective is provided in Figure~\ref{fig:heur4}, where we fix $K=10$ and vary $m$ from $1$ to $9$. Since anisotropy is constant in $m$, this isolates the effect of shifting dimensions between PoI and nuisance blocks: $D_\infty$ increases with $m$, confirming that a larger number of constrained coordinates—whether PoI or nuisance—generally reduces accuracy. Yet even in the extreme case of $m=9$, the maximum discrepancy remains around $15\%$, a moderate deviation given the severity of the constraint structure.

Overall, the discrepancies remain limited in magnitude and largely predictable in direction: the approximation captures the global shape well but tends to be anti-conservative in the upper tail. Thus, the rank-based approach provides a practically useful analytic surrogate -- particularly in settings with several boundary nuisance parameters, where exact $\bar{\chi}^2$ weights are unavailable -- provided that its systematic tail behaviour is acknowledged and, when needed, corrected for via calibration.

\begin{figure}[H]
    \centering
    \includegraphics[width=0.32\linewidth]{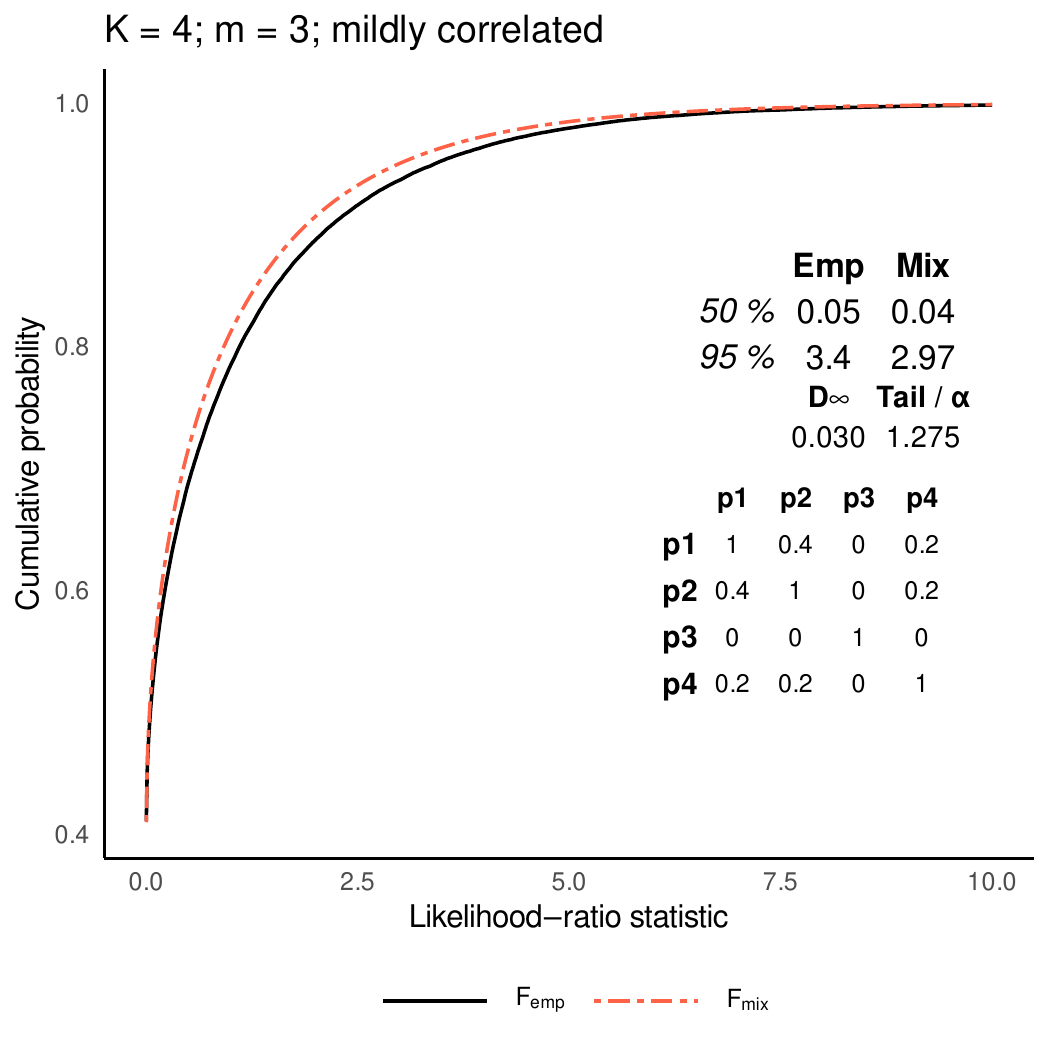}
    \includegraphics[width=0.32\linewidth]{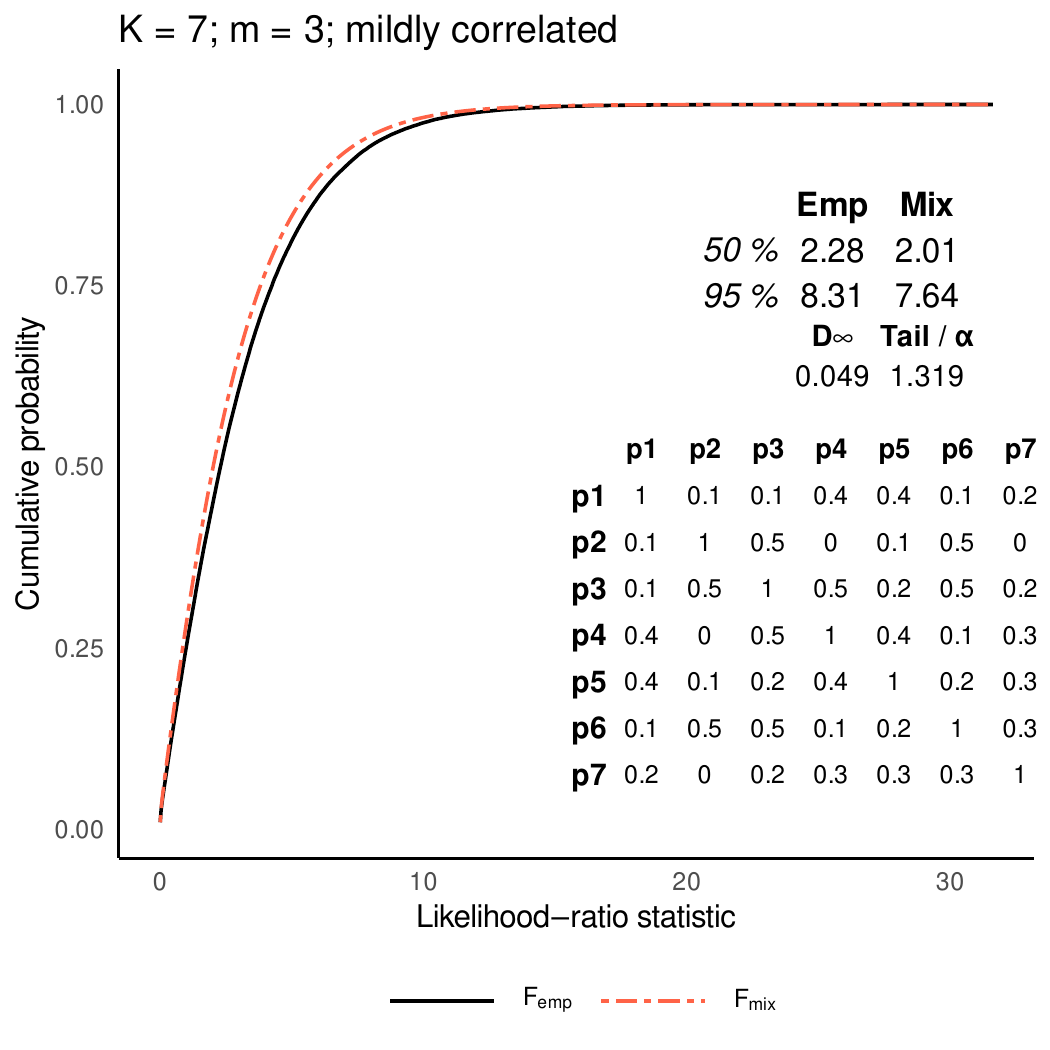}
    \includegraphics[width=0.32\linewidth]{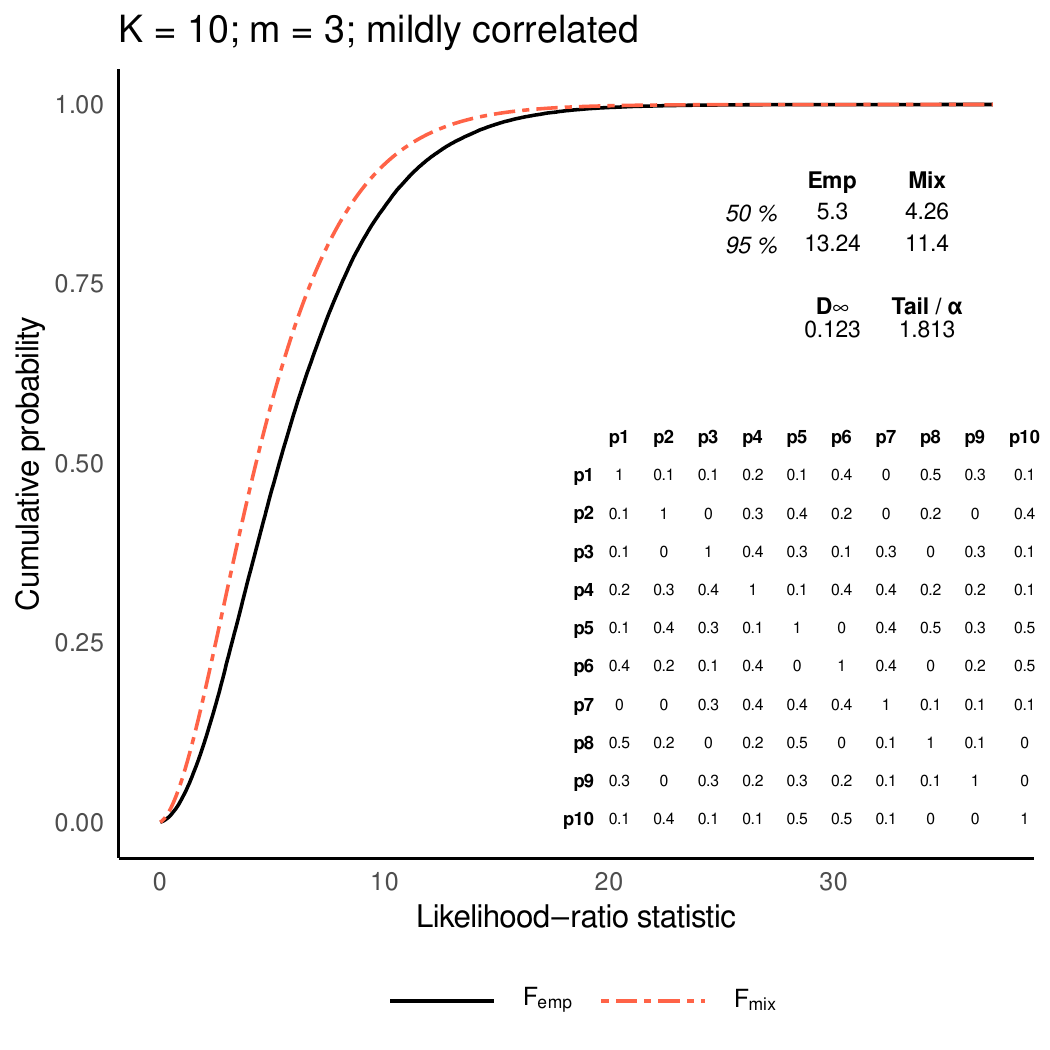}
    \caption{Empirical (solid black) versus approximated theoretical (dashed red) CDFs of $\lrs$ for $K=4,7,10$ in the mildly correlated case with three nuisance parameters on the boundary.}
    \label{fig:heur1}
\end{figure}

\begin{figure}[H]
    \centering
    \includegraphics[width=0.32\linewidth]{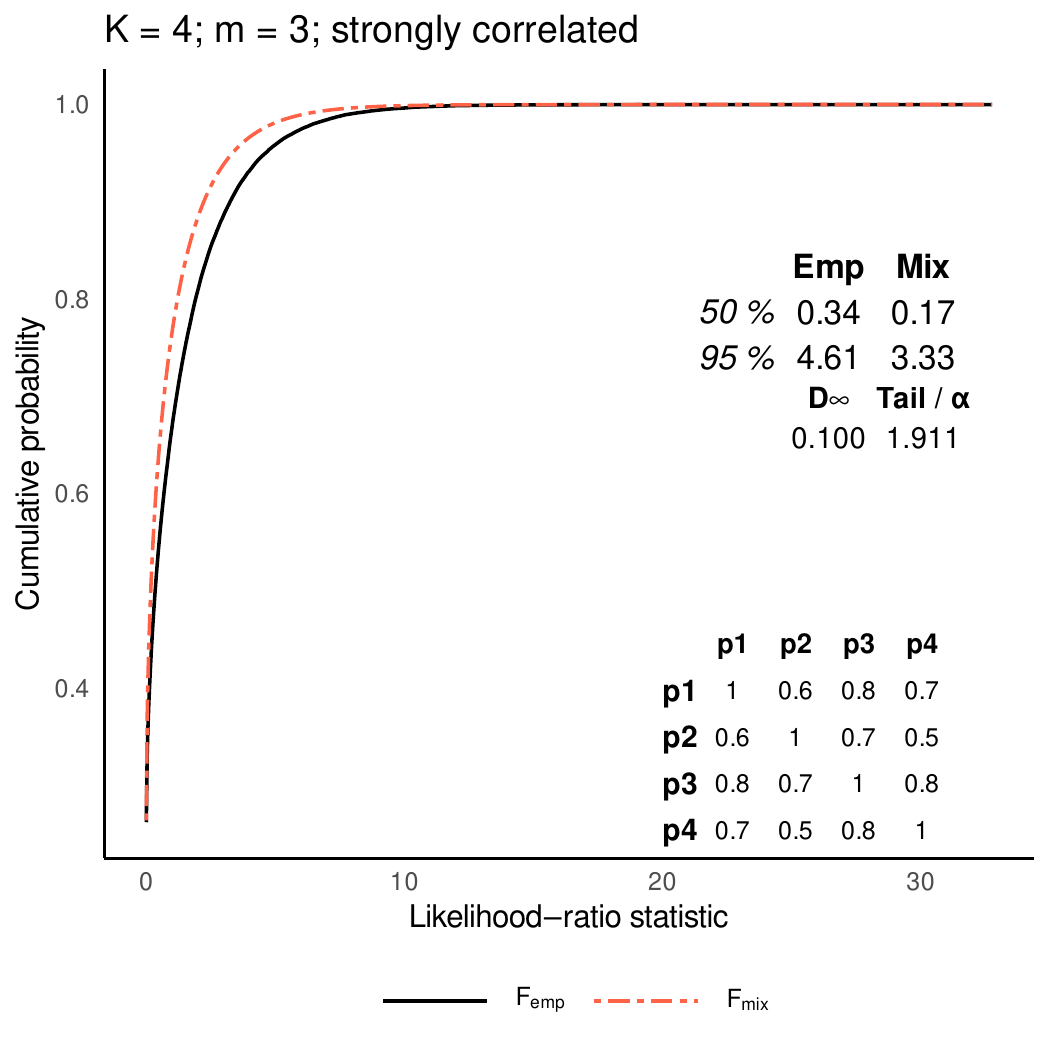}
    \includegraphics[width=0.32\linewidth]{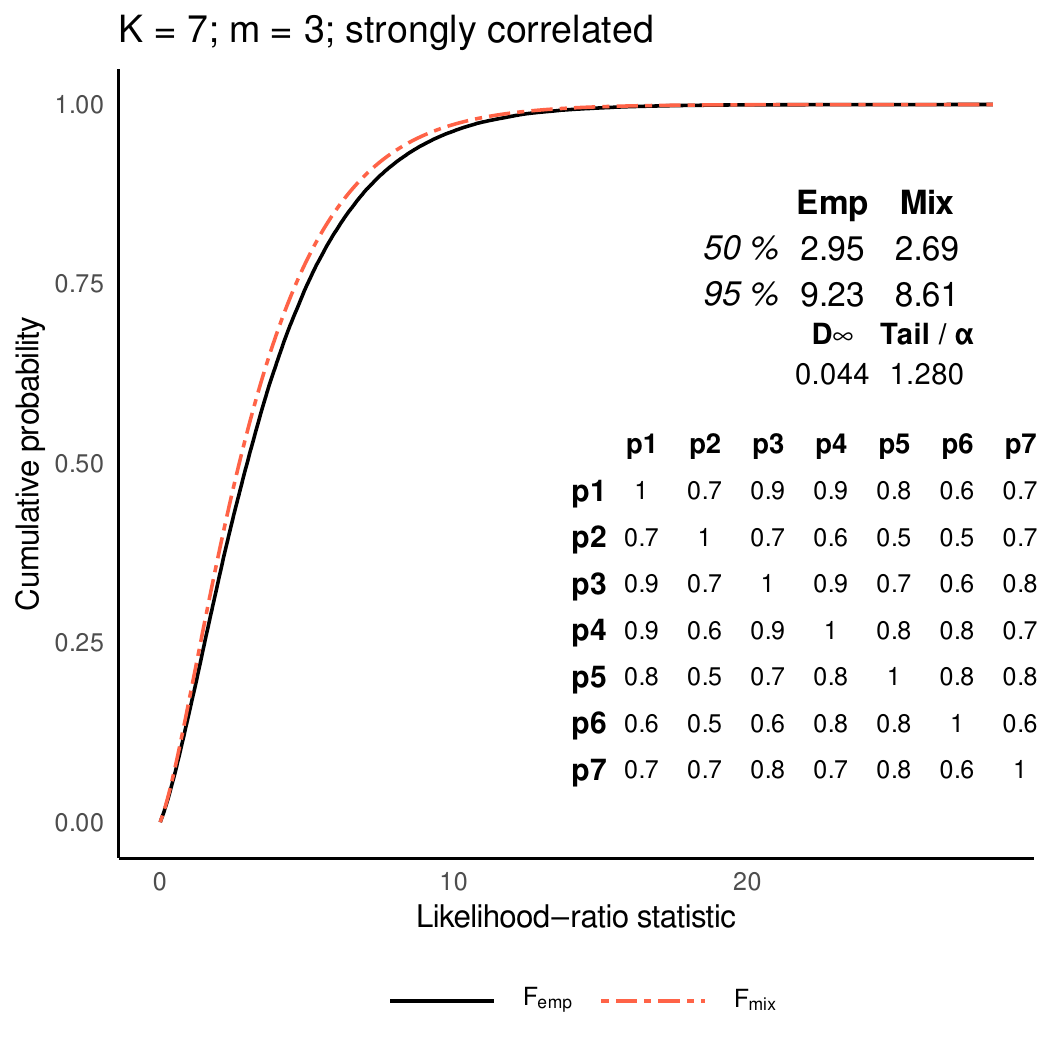}
    \includegraphics[width=0.32\linewidth]{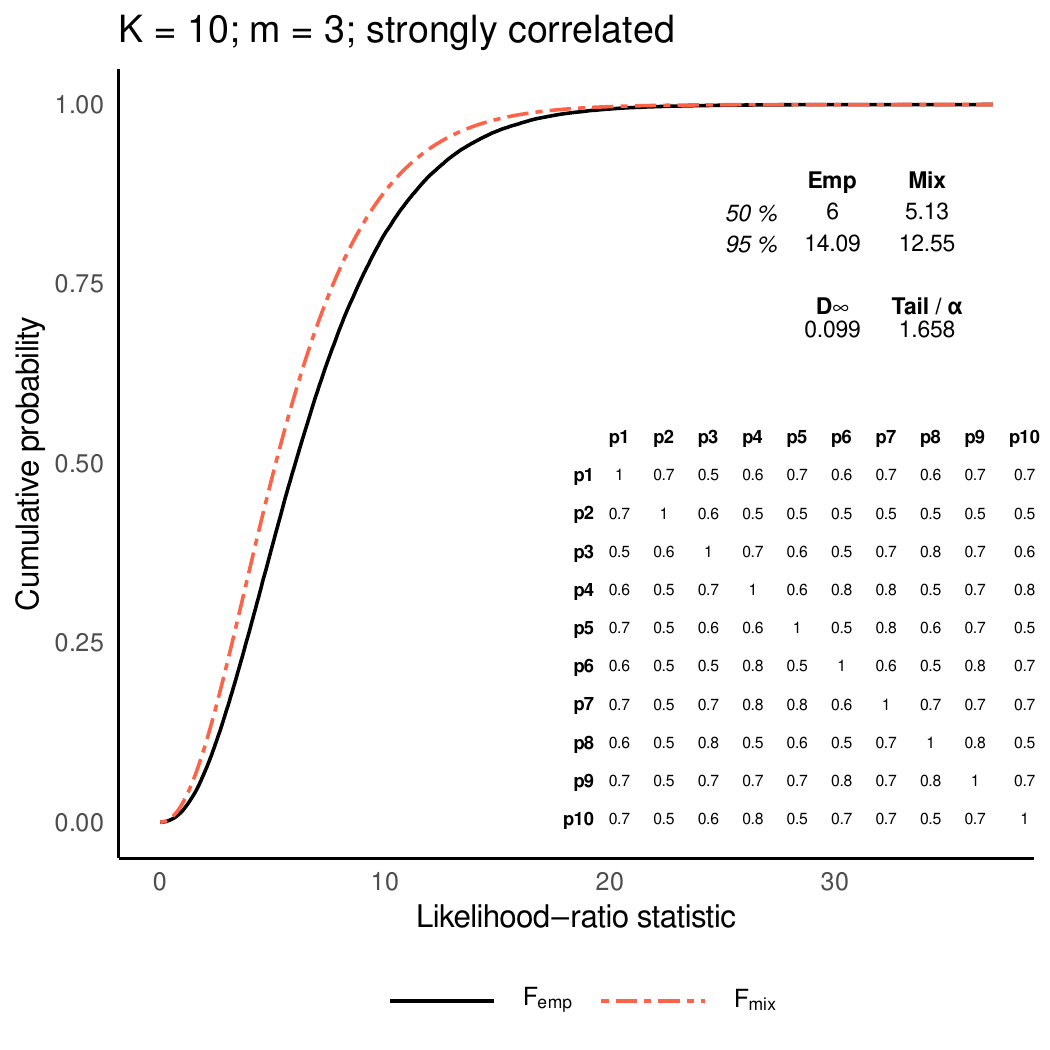}
    \caption{Empirical (solid black) versus approximated theoretical (dashed red) CDFs of $\lrs$ for $K=4,7,10$ in the strongly correlated case with three nuisance parameters on the boundary.}
    \label{fig:heur2}
\end{figure}

\begin{figure}[H]
    \centering
    
    \begin{subfigure}[t]{0.48\linewidth}
        \centering
        \includegraphics[width=\linewidth]{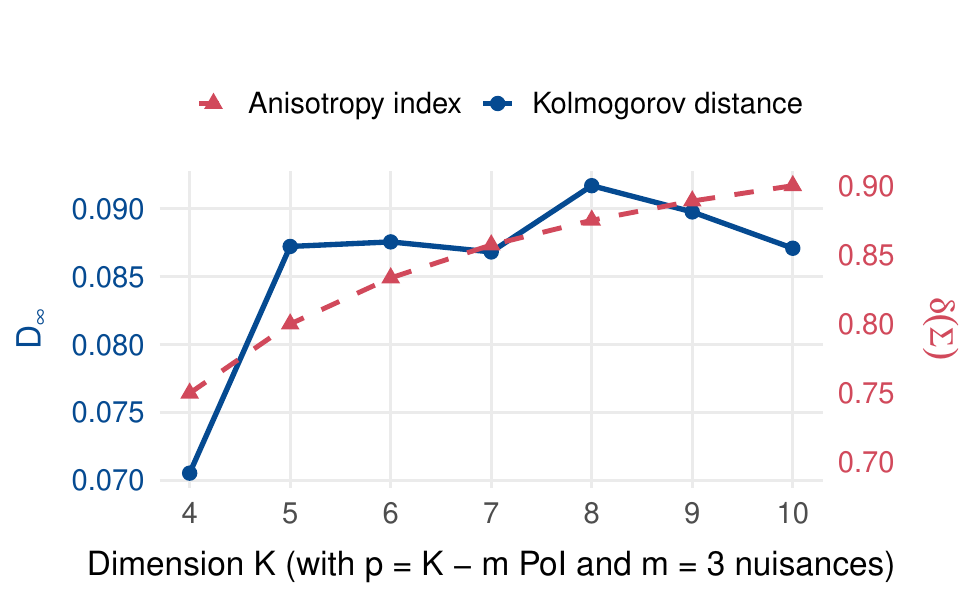}
        \caption{Kolmogorov distance $D_\infty$ (blue) and anisotropy index 
        $\delta(\Sigma_\rho)$ (red) as functions of $K$ for $\Sigma_\rho$ with 
        equicorrelated covariance ($\rho = 0.5$) in the case of $K-3$ PoIs and 
        3 nuisance parameters on the boundary.}
        \label{fig:heur3}
    \end{subfigure}
    \hfill
    \begin{subfigure}[t]{0.48\linewidth}
        \centering
        \includegraphics[width=\linewidth]{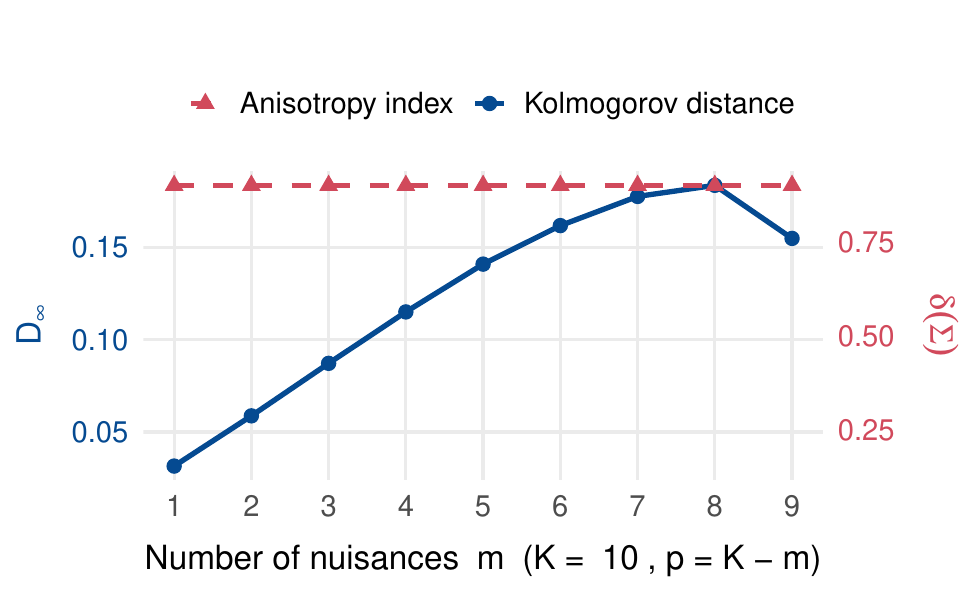}
        \caption{Kolmogorov distance $D_\infty$ (blue) and anisotropy index 
        $\delta(\Sigma_\rho)$ (red) as functions of the number of nuisance 
        parameters $m$ for $\Sigma_\rho$ with equicorrelated covariance 
        ($\rho = 0.5$), with $p = K - m$ PoIs and fixed total dimension $K=10$.}
        \label{fig:heur4}
    \end{subfigure}

    \caption{Comparison of Kolmogorov distance and anisotropy index across 
    two configurations.}
\end{figure}

\section*{Acknowledgements}
This work has been supported by the Research Council of Norway (RCN) through the FRIPRO PLUMBIN' (grant no.\ 323985) and the Centre of Excellence INTEGREAT (grant no.\ 332645). The author extends sincere gratitude to Riccardo De Bin for the insightful feedback and sustained support, and to Paolo Bertinelli for many interesting discussions.

\bibliography{paper-ref}

@article{Chernoff1954,
author = {Chernoff, H.},
title = {On the distribution of the likelihood ratio},
volume = {25},
journal = {The Annals of Mathematical Statistics},
pages = {573 -- 578},
year = {1954},
doi = {10.1214/aoms/1177728725}
}

@article{Bartholomew,
 author = {D. J. Bartholomew},
 journal = {Journal of the Royal Statistical Society. Series B (Methodological)},
 pages = {239--281},
 title = {A Test of Homogeneity of Means Under Restricted Alternatives},
 volume = {23},
 year = {1961}
}

@article{SelfLiang1987,
 author = {Self, S.G. and Liang, K.-Y.},
 journal = {Journal of the American Statistical Association},
 pages = {605-610},
 title = {Asymptotic Properties of Maximum Likelihood Estimators and Likelihood Ratio Tests Under Nonstandard Conditions},
 volume = {82},
 year = {1987}
}

@article{KS,
 author = {Kopylev, L. and Sinha, B.},
 journal = {Sankhya: The Indian Journal of Statistics, Series B},
 pages = {20-41},
 title = {On the asymptotic distribution of likelihood ratio test when parameters lie on the boundary},
 volume = {73},
 year = {2011}
}

@article{kuriki_takemura,
 author = {Takemura, A. and Kuriki, S.},
 journal = {{The Annals of Statistics}},
 pages = {2368--2387},
 title = {Weights of $\bar \chi^2$ Distribution for Smooth or Piecewise Smooth Cone Alternatives},
 volume = {25},
 year = {1997}
}

@article{amelunxen2014,
  title        = {Living on the Edge: Phase Transitions in Convex Programs with Random Data},
  author       = {Amelunxen, Dennis and Lotz, Martin and McCoy, Michael B. and Tropp, Joel A.},
  journal      = {arXiv:1303.6672},
  year         = {2014},
  archivePrefix= {arXiv},
  eprint       = {1303.6672},
  primaryClass = {cs.IT}
}

@article{us,
  title        = {A note on the asymptotic distribution of the Likelihood Ratio Test statistic under boundary conditions},
  author       = {Bertinelli Salucci, C. and Kvellestad, A. and De Bin, R.},
  journal      = {arXiv:2509.00223},
  year         = {2025},
  archivePrefix= {arXiv},
  eprint       = {2509.00223},
  primaryClass = {math.ST}
}

@article{hep,
  title        = {Under-coverage in high-statistics counting experiments with finite MC samples},
  author       = {Alexe, C.-A.  and Bendavid, J. and Bianchini , L. and Bruschini, D.},
  journal      = {arXiv:2401.10542},
  year         = {2025},
  archivePrefix= {arXiv},
  eprint       = {2401.10542},
  primaryClass = {physics.data-an}
}

@article{AmelunxenLotz2017,
  author    = {Amelunxen, D. and Lotz, M.},
  title     = {Intrinsic Volumes of Polyhedral Cones: {A} Combinatorial Perspective},
  journal   = {Discrete \& Computational Geometry},
  volume    = {58},
  pages     = {371--409},
  year      = {2017},
  doi       = {10.1007/s00454-017-9906-y}
}

@article{Kudo1963,
  author    = {Kudo, A.},
  title     = {A Multivariate Analogue of the One-Sided Test},
  journal   = {Biometrika},
  volume    = {50},
  pages     = {403--418},
  year      = {1963},
  doi       = {10.1093/biomet/50.3-4.403}
}

@article{LeCam1960,
  author    = {Le Cam, Lucien},
  title     = {Locally Asymptotically Normal Families of Distributions},
  journal   = {University of California Publications in Statistics},
  volume    = {3},
  pages     = {37--98},
  year      = {1960}
}

@book{vanDerVaart1998, 
series={Cambridge Series in Statistical and Probabilistic Mathematics}, 
title={Asymptotic Statistics}, 
publisher={Cambridge University Press}, 
author={van der Vaart, A. W.},
year={1998}, 
collection={Cambridge Series in Statistical and Probabilistic Mathematics}}

@book{Robertson1988,
  author    = {Robertson, Tim and Wright, F. T. and Dykstra, R. L.},
  title     = {Order Restricted Statistical Inference},
  publisher = {John Wiley \& Sons},
  address   = {New York},
  year      = {1988}
}

@article{Wolak1987,
  author    = {Wolak, Frank A.},
  title     = {An exact test for multiple inequality and equality constraints in the linear regression model},
  journal   = {Journal of the American Statistical Association},
  volume    = {82},
  number    = {399},
  pages     = {782--793},
  year      = {1987},
  doi       = {10.2307/2289471}
}

@article{Shapiro1985,
  author    = {Shapiro, A.},
  title     = {Asymptotic distribution of test statistics in the analysis of moment structures under inequality constraints},
  journal   = {Biometrika},
  volume    = {72},
  pages     = {133--144},
  year      = {1985},
  doi       = {10.2307/2336343}
}

@article{Shapiro1988,
  author    = {Shapiro, Alexander},
  title     = {Towards a unified theory of inequality constrained testing in multivariate analysis},
  journal   = {International Statistical Review},
  volume    = {56},
  number    = {1},
  pages     = {49--62},
  year      = {1988},
  doi       = {10.2307/1403367}
}

@article{Sun1988,
  author    = {Sun, Dong},
  title     = {A general reduction method for the computation of multivariate normal orthant probabilities},
  journal   = {Communications in Statistics – Simulation and Computation},
  volume    = {17},
  number    = {3},
  pages     = {777--791},
  year      = {1988},
  doi       = {10.1080/03610918808812636}
}

@article{Silvapulle1996,
  author    = {Silvapulle, Mervyn J.},
  title     = {A test in the linear model and chi-bar-squared distribution},
  journal   = {Journal of the American Statistical Association},
  volume    = {91},
  number    = {434},
  pages     = {545--551},
  year      = {1996},
  doi       = {10.1080/01621459.1996.10476700}
}

@article{LinLindsay1997,
  author    = {Lin, Liansheng and Lindsay, Bruce G.},
  title     = {The geometric approach to the asymptotic distribution of likelihood ratio tests under order restrictions},
  journal   = {Annals of Statistics},
  volume    = {25},
  number    = {1},
  pages     = {1--17},
  year      = {1997},
  doi       = {10.1214/aos/1034276628}
}

@article{Miwa2003,
  author    = {Miwa, Tatsuhiko and Hayter, Anthony J. and Kuriki, Satoshi},
  title     = {The evaluation of the multivariate normal probabilities with inequality constraints},
  journal   = {Journal of the Royal Statistical Society: Series B (Statistical Methodology)},
  volume    = {65},
  number    = {1},
  pages     = {223--234},
  year      = {2003},
  doi       = {10.1111/1467-9868.00381}
}

@book{GenzBretz2009,
  author    = {Genz, Alan and Bretz, Frank},
  title     = {Computation of Multivariate Normal and t Probabilities},
  publisher = {Springer},
  address   = {Heidelberg},
  series    = {Lecture Notes in Statistics},
  volume    = {195},
  year      = {2009},
  doi       = {10.1007/978-3-642-01689-9}
}

@article{econometrics,
author = {Wu, B. and Yao, Q. and Zhu, S.},
year = {2013},
pages = {2877-2898},
title = {Estimation in the presence of many nuisance parameters: Composite likelihood and plug-in likelihood},
volume = {123},
journal = {{Stochastic Processes and their Applications}},
doi = {10.1016/j.spa.2013.03.017}
}

@article{cosmology,
    author = {Kitching, T. D. and Amara, A. and Abdalla, F. B. and Joachimi, B. and Refregier, A.},
    title = {Cosmological systematics beyond nuisance parameters: form-filling functions},
    journal = {{Monthly Notices of the Royal Astronomical Society}},
    volume = {399},
    pages = {2107-2128},
    year = {2009},
    doi = {10.1111/j.1365-2966.2009.15408.x}
}

@article{biostat,
  author  = {Wu, C. and Xu, G. and Shen, X. and Pan, W.},
  title   = {A Regularization-Based Adaptive Test for High-Dimensional {GLM}s},
  journal = {{Journal of Machine Learning Research}},
  year    = {2020},
  volume  = {21},
  pages   = {1-67}
}

\appendix

\section*{Appendix 1} \label{appA}
\subsection*{Proofs of Lemma 1 and Lemma 2}  

\noindent \textbf{Lemma 1}

\begin{proof}
After the linear transformation described in Section~\ref{sec:background},  
we have $\tilde Z=(\tilde Z_1,\ldots,\tilde Z_K)^\top\sim N_K(0,\mathbb{I}_K)$.  
For the null cone $\tilde C^{\mathrm{pt}}_0=\{0\}$, i.e.\ when all $K$ parameters are of interest,  
the projection $P_{\tilde C}(\tilde Z)$ keeps the positive coordinates of $\tilde Z$  
and sets to zero the negative ones. Hence, the face on which  
$P_{\tilde C}(\tilde Z)$ lies has dimension equal to the number of positive coordinates,  
$j=\#\{i:\tilde Z_i>0\}$. Therefore, since the coordinates of $\tilde Z$ are independent and symmetric, we have
\(
\dot w_j^\perp= \Pr \left(\#\{i:\tilde Z_i>0\}=j\right)
          = 2^{-K}\binom{K}{j}
\) for $j=0,\ldots,K$.

Now consider the case where the last of the $K$ parameters is treated as nuisance, so that the null cone becomes $\tilde C^{\mathrm{ray}}_0=\mathbb{R}_+ e_K=\{(0,\ldots,0,t):t\ge 0\}$.  
The orthogonal projection onto $\tilde C^{\mathrm{ray}}_0$ is then  
$P_{\tilde C_0}(\tilde Z)=\bigl(0,\ldots,0,(\tilde Z_K)_+\bigr)$,  
where $(x)_+=\max\{x,0\}$, but the difference in squared projection norms eliminates the contribution of the nuisance coordinate, since
\[
\lrs = \|P_{\tilde C}(\tilde Z)\|^2-\|P_{\tilde C_0}(\tilde Z)\|^2
  = \sum_{i=1}^{K} (\tilde Z_i^+)^2 - (\tilde Z_K^+)^2
  = \sum_{i=1}^{K-1} (\tilde Z_i^+)^2.
\]
The relevant face of $\tilde C$ determining the projection therefore depends only on which of the first $K-1$ coordinates are positive:
its dimension is thus $j=\#\{i\le K-1:\tilde Z_i>0\}$, irrespective of the value of $\tilde Z_K$.

Again, by independence and symmetry, we have
\(
\bar w_j^\perp
 = \Pr \bigl(\#\{i\le K-1:\tilde Z_i>0\}=j\bigr)
 = 2^{-(K-1)}\binom{K-1}{j} \) for \(j=0,\ldots,K,
\)
with $\binom{K-1}{K}=0$ as, by definition, $\binom{n}{k}=0$ for $k>n$.  

\smallskip

Define the differences $\Delta_j^\perp= \bar w_j^\perp - \dot w_j^\perp$.
Using Pascal’s identity $\binom{K}{j}=\binom{K-1}{j}+\binom{K-1}{j-1}$,
for $1\le j\le K-1$ we obtain
\[
\Delta^\perp_j
= 2^{-(K-1)}\binom{K-1}{j} \;-\; 2^{-K}\binom{K}{j}
= 2^{-K}\!\left[2\binom{K-1}{j}-\binom{K-1}{j}-\binom{K-1}{j-1}\right]
= 2^{-K}\!\left[\binom{K-1}{j}-\binom{K-1}{j-1}\right].
\]
For the indices at the extremes,
\[
\Delta^\perp_0
= 2^{-(K-1)}\binom{K-1}{0} - 2^{-K}\binom{K}{0}
= 2^{-(K-1)} - 2^{-K}
= 2^{-K},
\]
and, since $\binom{n}{k}=0$ for $k>n$,
\[
\Delta^\perp_K
= 2^{-(K-1)}\binom{K-1}{K} - 2^{-K}\binom{K}{K}
= 0 - 2^{-K}
= -\,2^{-K}.
\]
Thus, the closed-form expression is proved for all $j=0,1,\ldots,K$, and it follows
\(
\bar w^\perp_j = \dot w^\perp_j + \Delta^\perp_j,
\) for $j=0,\ldots,K$.
As a check, $\sum_{j=0}^{K}\Delta^\perp_j
= \sum_{j=0}^{K}\bar w^\perp_j - \sum_{j=0}^{K} \dot w^\perp_j
= 1-1=0$, confirming that the transformation merely redistributes the mixture mass.
\end{proof}

\vspace{.5cm}

\noindent \textbf{Lemma 2}

\begin{proof}
The argument is the same as in the proof of Lemma~\ref{lem:delta-weights}, applied to the first $K-m$ coordinates. After the linear transformation in Section~\ref{sec:background},
$\tilde Z=(\tilde Z_1,\ldots,\tilde Z_K)^\top\sim N_K(0,\mathbb{I}_K)$
with independent symmetric coordinates. For the point-null cone $C^{\mathrm{pt}}_0=\{0\}$ the weights are as before,
\(
\dot w^\perp_j = 2^{-K}\binom{K}{j}
\). Now let the last $m$ coordinates be nuisance, so that
$\tilde C^{(m)}_0 = \mathbb{R}_+^m$ in the last $m$ axes. The projection
onto $\tilde C^{(m)}_0$ is
\(
P_{\tilde C^{(m)}_0}(\tilde Z)
  = \bigl(0,\ldots,0,(\tilde Z_{K-m+1})_+,\ldots,(\tilde Z_K)_+\bigr),
\)
and 
\(
\lrs
  = \|P_{\tilde C}(\tilde Z)\|^2 - \|P_{\tilde C^{(m)}_0}(\tilde Z)\|^2
  = \sum_{i=1}^{K} (\tilde Z_i^+)^2
      \;-\;\sum_{i=K-m+1}^{K} (\tilde Z_i^+)^2
  = \sum_{i=1}^{K-m} (\tilde Z_i^+)^2
\). Thus the relevant face of $\tilde C$ depends only on the signs of the
first $K-m$ coordinates, \(
j = \#\{i \le K-m : \tilde Z_i>0\}
\). Independence and symmetry give
\[
\bar w^{(m)\perp}_j
  = \Pr\!\left(\#\{i\le K-m : \tilde Z_i>0\}=j\right)
  = 2^{-(K-m)} \binom{K-m}{j},
  \qquad j=0,\ldots,K-m,
\]
and $\bar w^{(m)\perp}_j=0$ for $j>K-m$.

Subtracting $\dot w^{\perp}_j$ gives the stated formulas for $\Delta^{(m)\perp}_j$.
Finally,
\[
\sum_{j=0}^{K}\bar w^{(m)\perp}_j
  = \sum_{j=0}^{K-m} 2^{-(K-m)} \binom{K-m}{j}
  = 1,
\qquad
\sum_{j=0}^{K}\Delta^{(m)\perp}_j = 0,
\]
so the transformed sequence is a valid set of $\bar{\chi}^2$ weights.
\end{proof}

\section*{Appendix 2} \label{appB}
\subsection*{Numerical validation of Lemma 1 and Lemma 2}

\noindent We first validate Lemma~\ref{lem:delta-weights}, which gives the exact change in $\bar{\chi}^2$ weights when one of $K$ orthogonal
parameters is demoted from PoI to nuisance parameter on the boundary. For $K \in \{4,7,10\}$, we consider the orthogonal case $\Sigma = \mathbb{I}_K$ and compare the empirical distribution of $\lrs$ to the $\bar{\chi}^2$ mixture with weights
\(
\bar{w}_j = \dot{w}_j + \Delta^\perp_j \) with
\(\dot{w}_j = 2^{-K}\binom{K}{j},
\)
where $\Delta^\perp_j$ is given by Lemma~\ref{lem:delta-weights}.  We see in Figure~\ref{fig:lemma1} that, in all configurations, the empirical CDF is indistinguishable, up to Monte Carlo error, from the theoretical mixture, and the quantile differences are negligible.  
This confirms the correctness of the closed-form redistribution pattern
\eqref{eq:deltaj} in the orthogonal case, and illustrates the intuitive picture that probability mass is shifted symmetrically across adjacent degrees of freedom when a single parameter is demoted to nuisance.

\vspace{0.5cm}

To validate Lemma~\ref{lem:delta-weights-m}, we now increase the number of nuisance parameters from $m=1$ to $m=3$ and compare the distribution of $\lrs$ to the analytic $\bar{\chi}^2$ mixture with weights given by the lemma. The exactness of the weight differences is confirmed by Figure~\ref{fig:lemma2} where in all configurations the analytic curve provides a perfect match, up to MC error, to the empirical CDF.  

\begin{figure}[H]
    \centering
    \includegraphics[width=0.32\linewidth]{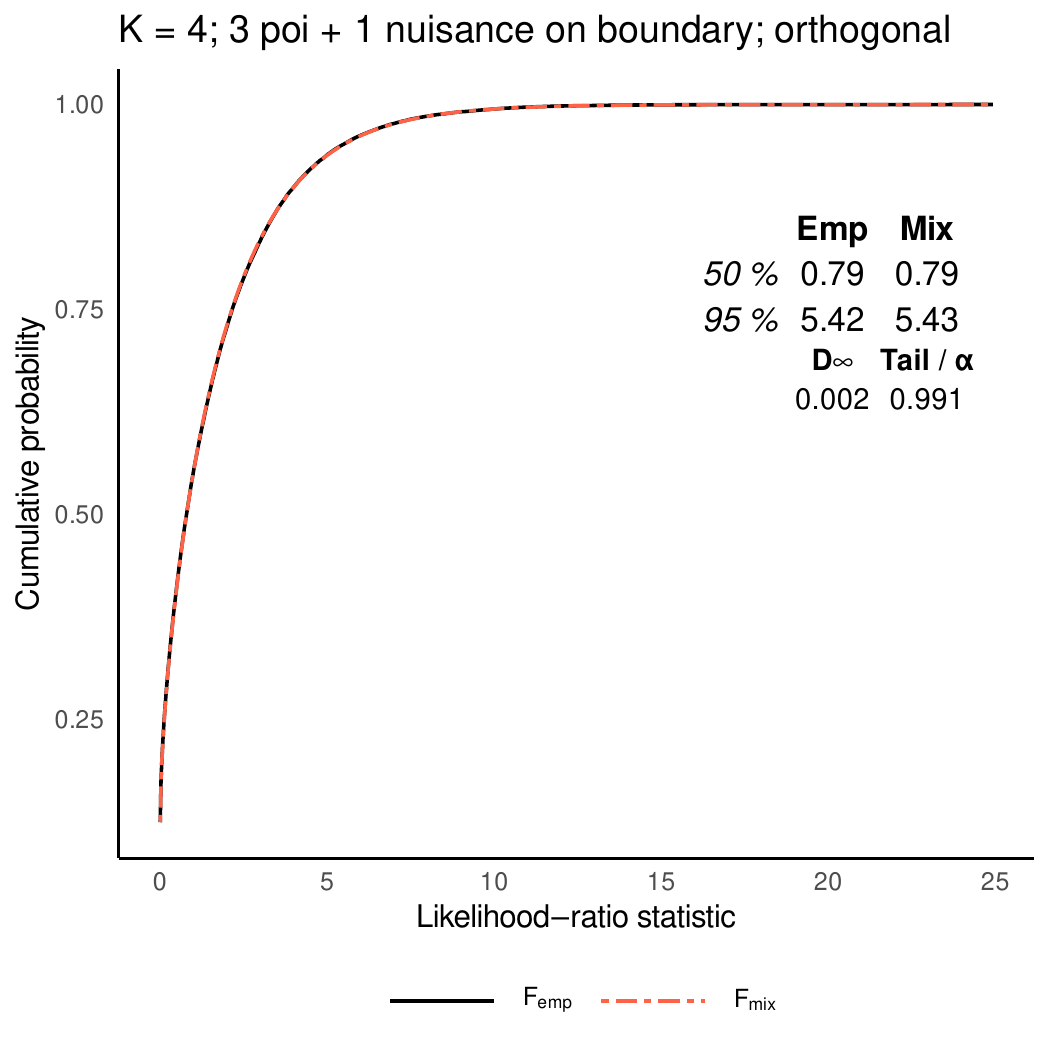}
    \includegraphics[width=0.32\linewidth]{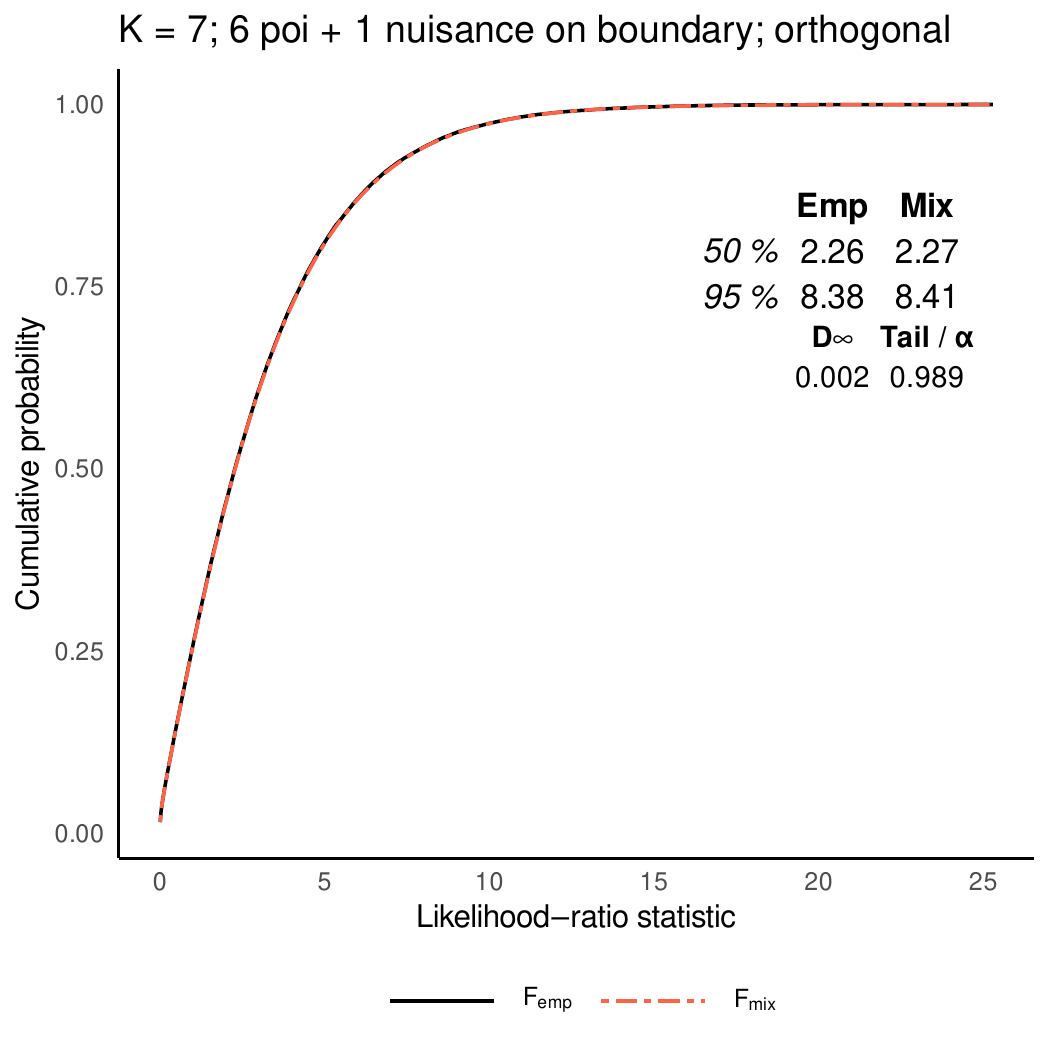}
    \includegraphics[width=0.32\linewidth]{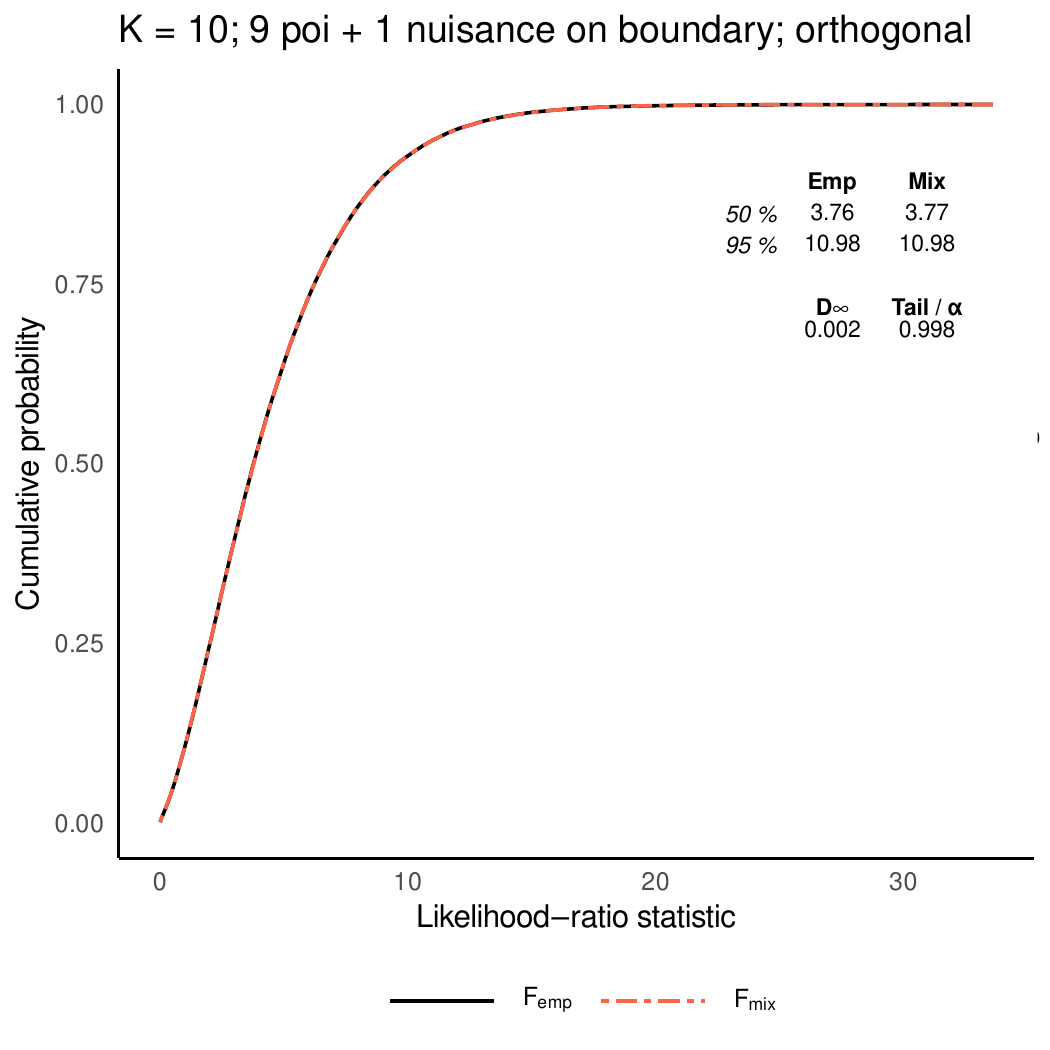}
    \caption{Empirical (solid black) versus theoretical (dashed red) CDFs of $\lrs$ for $K=4,7,10$ in the orthogonal case with one nuisance parameter on the boundary.}
    \label{fig:lemma1}
\end{figure}

\begin{figure}[H]
    \centering
    \includegraphics[width=0.32\linewidth]{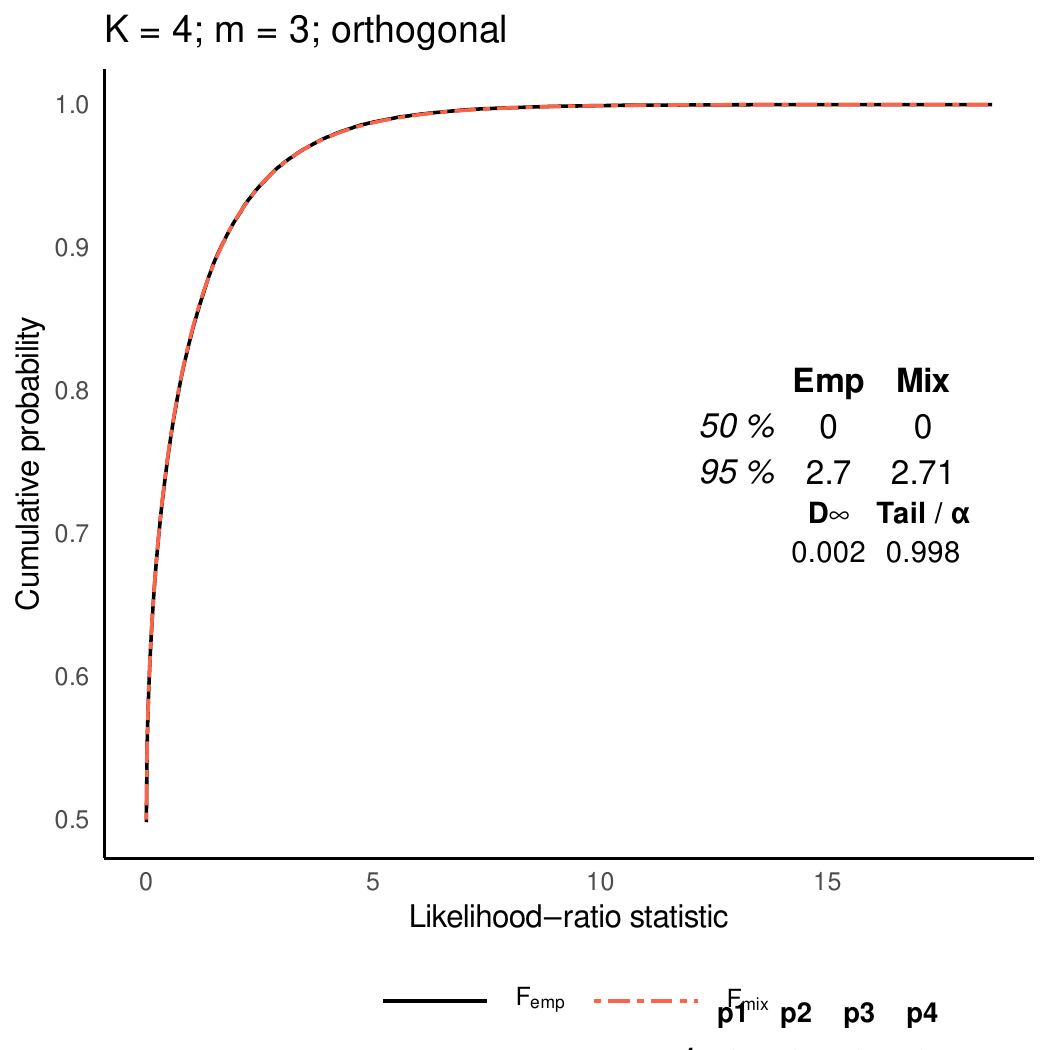}
    \includegraphics[width=0.32\linewidth]{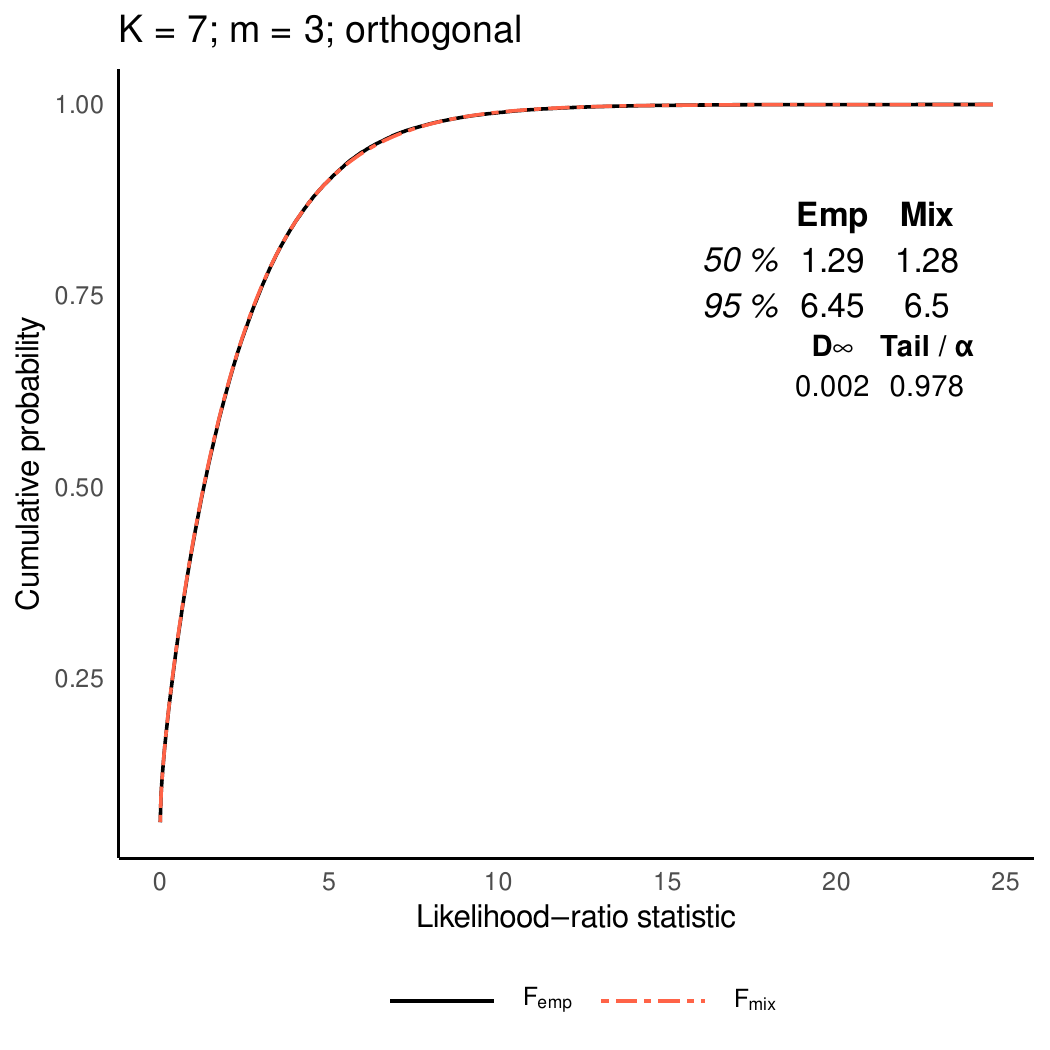}
    \includegraphics[width=0.32\linewidth]{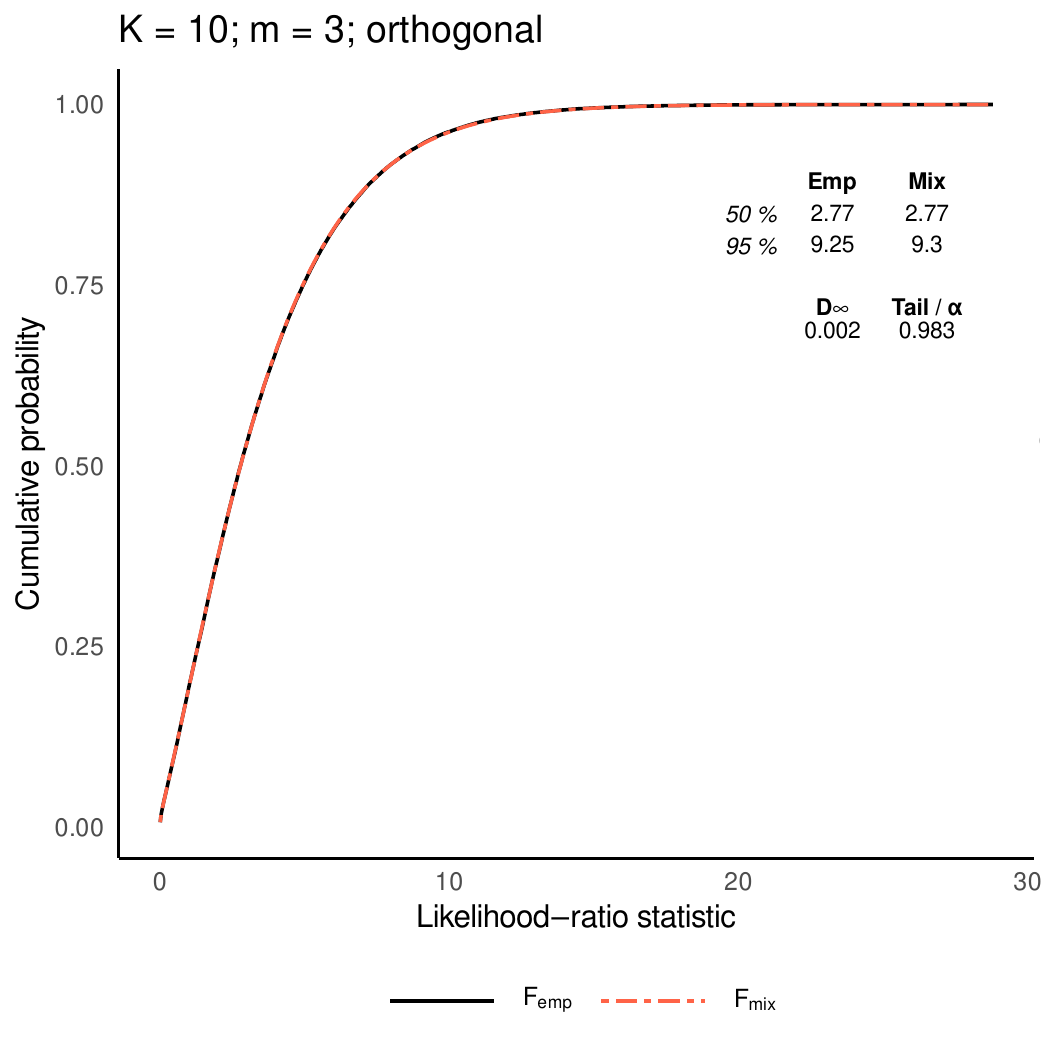}
    \caption{Empirical (solid black) versus theoretical (dashed red) CDFs of $\lrs$ for $K=4,7,10$ in the orthogonal case with three nuisance parameters on the boundary.}
    \label{fig:lemma2}
\end{figure}

\end{document}